\documentclass[journal]{IEEEtran}
\usepackage{amsmath}
\usepackage{amssymb}
\usepackage{amsfonts}
\usepackage{graphicx}
\usepackage{epsfig}
\usepackage{subfigure}
\usepackage{psfrag}
\usepackage{color}

\begin{document}
\title{Wireless Information and Power Transfer: A Dynamic Power Splitting Approach}

\author{Liang Liu, Rui Zhang,~\IEEEmembership{Member,~IEEE}, and Kee-Chaing Chua,~\IEEEmembership{Member,~IEEE}

\thanks{L. Liu is with the
Department of Electrical and Computer Engineering, National
University of Singapore (e-mail:liu\_liang@nus.edu.sg).}
\thanks{R. Zhang is with the Department of Electrical and Computer Engineering, National
University of Singapore (e-mail: elezhang@nus.edu.sg). He is also
with the Institute for Infocomm Research, A*STAR, Singapore.}
\thanks{K. C. Chua is with the
Department of Electrical and Computer Engineering, National
University of Singapore (e-mail:eleckc@nus.edu.sg).}}

\maketitle

\begin{abstract}
Energy harvesting is a promising solution to prolong the operation time
of energy-constrained wireless networks. In particular, scavenging
energy from ambient radio signals, namely \emph{wireless energy
harvesting} (WEH), has recently drawn significant attention. In this
paper, we consider a point-to-point wireless link over the flat-fading channel, where the receiver has no
fixed power supplies and thus needs to replenish energy via WEH from the signals sent by the transmitter. We first
consider a SISO (single-input single-output) system where the single-antenna receiver cannot decode information and
harvest energy independently from the same signal received. Under this practical constraint, we propose a {\it dynamic power splitting} (DPS) scheme, where the
received signal is split into two streams with adjustable power levels for information decoding and energy harvesting separately
based on the instantaneous channel condition that is assumed to be known at the receiver. We derive the optimal power splitting rule at the receiver to achieve various trade-offs
between the maximum ergodic capacity for information transfer and the
maximum average harvested energy for power transfer, which are characterized by
the boundary of a so-called ``rate-energy'' region. Moreover, for
the case when the channel state information is also known at the
transmitter, we investigate the joint optimization of transmitter power
control and
receiver power splitting. The achievable rate-energy (R-E) region by the proposed DPS scheme is compared against that by the existing time switching scheme as well as a performance upper bound by ignoring the practical receiver constraint. Finally, we extend the result for DPS to the SIMO (single-input multiple-output) system where the receiver is equipped with multiple antennas. In particular, we investigate a low-complexity power splitting scheme, namely \emph{antenna switching}, which can be practically implemented to achieve the near-optimal rate-energy trade-offs as compared to the optimal DPS.
\end{abstract}


\setlength{\baselineskip}{1.0\baselineskip}
\newtheorem{definition}{\underline{Definition}}[section]
\newtheorem{fact}{Fact}
\newtheorem{assumption}{Assumption}
\newtheorem{theorem}{\underline{Theorem}}[section]
\newtheorem{lemma}{\underline{Lemma}}[section]
\newtheorem{corollary}{\underline{Corollary}}[section]
\newtheorem{proposition}{\underline{Proposition}}[section]
\newtheorem{example}{\underline{Example}}[section]
\newtheorem{remark}{\underline{Remark}}[section]
\newtheorem{algorithm}{\underline{Algorithm}}[section]
\newcommand{\mv}[1]{\mbox{\boldmath{$ #1 $}}}

\section{Introduction}

\PARstart{R}ecently,
energy harvesting has become a prominent solution to prolong the
lifetime of energy-constrained wireless networks, such as sensor networks. Compared with conventional energy supplies such as batteries that have fixed operation time, energy harvesting from the environment potentially provides an unlimited energy supply for wireless networks. Besides other commonly used
energy sources such as solar and wind, radio frequency (RF) signal holds a promising future for wireless energy harvesting (WEH) since it can also be used to provide wireless information transmission at the same time, which has motivated an upsurge of research interest on RF-based wireless
information and power transfer recently \cite{Varshney08}-\cite{Rui12}. Prior works \cite{Varshney08}, \cite{Sahai10} have studied the
fundamental performance limits of wireless systems with simultaneous
information and power transfer, where the
receiver is ideally assumed to be able to decode the information and harvest
the energy independently from the same received signal. However,
this assumption implies that the received signal used for
harvesting energy can be reused for decoding information without any loss, which is not realizable yet due to practical circuit limitations.
Consequently, in \cite{Rui11} the authors proposed two practical
receiver designs, namely ``time switching'', where the receiver switches
between decoding information and harvesting energy at any time, and
``power splitting'', where the receiver splits the signal into two streams of different power for
decoding information and harvesting energy separately, to enable
WEH with simultaneous information transmission.

In this paper, we further investigate the power splitting scheme in
\cite{Rui11} for a point-to-point single-antenna flat-fading channel, where the receiver is able to dynamically adjust the split power ratio for information decoding and energy harvesting based on the channel state information (CSI) that is assumed to be known at the receiver, a scheme so-called ``dynamic power splitting (DPS)'' as shown in Fig. \ref{fig1}.
We assume that the transmitter has a constant power supply, whereas the receiver has no fixed power supplies and thus
needs to harvest energy from the received
signal sent by the transmitter. For the ease of hardware implementation, we consider the case where the information decoding circuit and energy harvesting circuit are separately designed (as opposed to an integrated design in \cite{Rui2012}). As a
result, the receiver needs to determine the amount of received signal power that is split to the information receiver versus that to the energy receiver
based on the instantaneous channel power. We
derive the optimal power splitting rule at the receiver to achieve
various tradeoffs between the maximum ergodic capacity for information transmission versus the maximum
average harvested energy for power transmission, which are characterized by the
boundary of a so-called
``rate-energy (R-E)'' region. Moreover, for the case of CSI also known at the transmitter (CSIT), we
examine the joint optimization of transmitter power control and receiver power splitting, and show the achievable R-E gains over the case without CSIT.

Furthermore, we extend the DPS scheme for the single-input single-output (SISO) system to the single-input multiple-output (SIMO) system, where the receiver is equipped with multiple antennas. After deriving the optimal DPS rule for the SIMO system which in general requires independent power splitters that are connected to different receiving antennas, we further investigate a low-complexity power splitting scheme so-called ``antenna switching'' proposed in \cite{Rui11}, whereby the total number of receiving antennas is divided into two subsets, one for decoding information and the other for  harvesting energy. It is noted that for the SISO fading channel case, antenna switching reduces to time switching, which has been studied in our previous work \cite{Rui12}. In \cite{Rui12}, the optimal time switching rule based on the receiver CSI and its corresponding transmitter power control policy (in the case of CSI known at the transmitter) were derived to achieve various trade-offs between wireless information and energy transfer. It was shown that for time switching, the optimal policy is threshold based, i.e., the receiver decodes information when the fading channel gain is below a certain threshold, and harvests energy otherwise. It is worth noting that although theoretically time switching can be regarded as a special form of power splitting with only on-off power allocation at each receiving antenna, they are implemented by different hardware circuits (time switcher versus power splitter) in practice.


The main results of this paper are summarized as follows:
\begin{itemize}
\item For the SISO case, we show that to achieve the optimal R-E trade-offs in both the cases without or with CSIT by DPS, a fixed amount of the received signal power should be allocated to the information receiver, with the remaining power allocated to the energy receiver when the fading channel gain is above a given threshold. However, when the fading channel gain is below this threshold, all the received power should be allocated to the information receiver. Compared with our previous result for the time switching receiver in \cite{Rui12} where only the energy harvesting receiver can benefit from ``good'' fading channels above the threshold, the DPS scheme utilizes the ``good'' fading states for both information decoding and energy harvesting. As a result, we show by simulations that DPS can achieve substantial R-E performance gains over dynamic time switching in the SISO fading channel. Moreover, we derive the R-E region for the ideal case when the receiver can decode information and harvest energy from the same received signal independently without any rate or energy loss, which provides a theoretical performance upper bound for the DPS scheme.
\item For the SIMO case where the receiver is equipped with multiple antennas, we extend the result for DPS as follows. First, we show that a uniform power splitting (UPS) scheme where all the receiving antennas are assigned with the same power splitting ratio is optimal. We derive  the optimal UPS rule and/or transmitter power control (in the case with CSIT) based on the result for the SISO system by treating all the receiving antennas as one virtual antenna with an equivalent channel sum-power. Second, to ease the hardware implementation of UPS, we investigate the optimal antenna switching rule to maximize the achievable R-E trade-offs. An exhaustive search algorithm is presented first, and then a new low-complexity antenna selection algorithm is proposed, which is shown to perform closer to the optimal UPS as the number of receiving antennas increases. Moreover, it is shown that with the optimal antenna selection, even with two receiving antennas, the R-E performance of antenna switching is already very close to that with the optimal UPS. This demonstrates the usefulness of antenna switching as a practically appealing low-complexity implementation for power splitting.

\end{itemize}

The rest of this paper is organized as follows. Section
\ref{System Model} presents the system
model and illustrates the encoding and decoding schemes for wireless
information transfer with opportunistic energy harvesting by DPS. Section
\ref{Ergodic Capacity and Harvested Energy Trade-off in Fading Channels} defines the R-E region achievable by DPS and formulates the
problems to characterize its boundaries without or with CSIT. Sections \ref{Rare-Energy Trade-off} presents the optimal DPS rule at the receiver and/or power control policy at the transmitter (in the case of CSIT) to achieve various R-E trade-offs in the SISO fading channel. Section \ref{sec:Implementation of Power Splitting by Multiple Receive Antennas} extends the result to the SIMO fading channel and investigates the practical scheme of antenna switching. Finally,
Section \ref{Concluding Remarks} concludes the paper.

\section{System Model}\label{System Model}

\begin{figure}
\begin{center}
\scalebox{0.45}{\includegraphics*{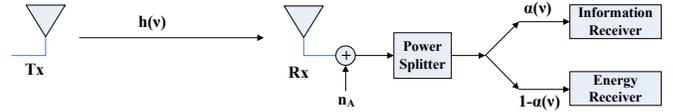}}
\end{center}
\caption{SISO system model.}\label{fig1}
\end{figure}

As shown in Fig. \ref{fig1}, we first consider a wireless SISO
link consisting of one pair of single-antenna
transmitter (Tx) and single-antenna receiver (Rx) over the flat-fading channel. The case of single-antenna Tx and multi-antenna Rx or SIMO system will be addressed later in Section \ref{sec:Implementation of Power Splitting by Multiple Receive Antennas}. For convenience, we assume that the channel from
Tx to Rx follows a block-fading model \cite{Shamai94}, \cite{Caire}. The equivalent complex baseband channel from Tx to Rx in
one particular fading state is denoted by $g(\nu)$, where $\nu$ denotes the fading state, and the channel power gain at fading state $\nu$ is denoted by $h(\nu)=|g(\nu)|^2$. It is assumed that the random variable (RV) $h(\nu)$ has a continuous probability density function (PDF) denoted by $f_\nu(h)$. At any fading state $\nu$,
$h(\nu)$ is assumed to be perfectly known at Rx, but may or may not be known at Tx.


We consider time-slotted transmissions at Tx and the DPS
scheme at Rx. As shown in Fig. \ref{fig1}, at Rx, the RF-band signal is corrupted by an additive noise $n_A$ introduced by the receiver antenna, which is assumed to be a circularly
symmetric complex Gaussian (CSCG) RV with zero mean and variance
$\sigma_A^2$, denoted by $n_A\sim \mathcal{CN}(0,\sigma_A^2)$, in its baseband equivalent. The RF-band signal is then fed into a power splitter \cite{Wu10}, \cite{datasheet}, where the signal plus the antenna noise is split to the information receiver and energy receiver \cite{Popovic08} separately. For each fading state $\nu$, the portion of signal power split to information decoding (ID) is
denoted by $\alpha(\nu)$ with $0\leq \alpha(\nu) \leq 1$, and that to energy harvesting (EH) as $1-\alpha(\nu)$, where in general $\alpha(\nu)$ can be adjusted over different fading states. The ID circuit introduces an additional baseband noise $n_{\rm ID}$ to the signal split to the information receiver, which is assumed to be a CSCG RV with zero mean and variance $\sigma^2$, and independent of the antenna noise $n_A$. As a result, the equivalent noise power for ID is $\alpha(\nu)\sigma_A^2+\sigma^2$ at fading state $\nu$. On the other hand, in addition to the split signal energy, the energy receiver can harvest $(1-\alpha(\nu))\sigma_A^2$ amount of energy (normalized by the slot duration) due to the antenna noise $n_A$. However, in practice, $n_A$ has a negligible influence on both the ID and EH since $\sigma_A^2$ is usually much smaller than the noise power introduced by the information receiver, $\sigma^2$, and thus even lower than the average power of the received signal. Thus, in the rest of this paper, we assume $\sigma_A^2=0$ for simplicity.

\begin{figure}
\begin{center}
\scalebox{0.45}{\includegraphics*{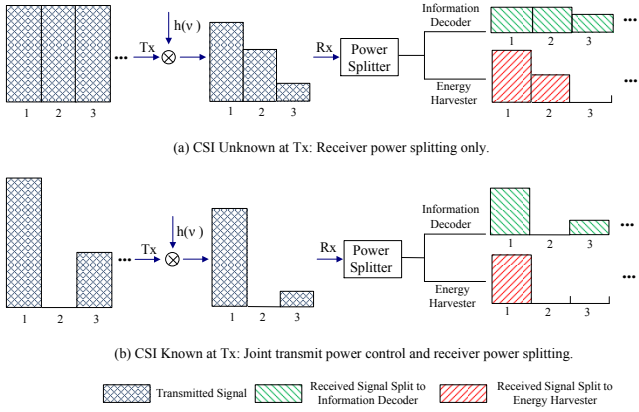}}
\end{center}
\caption{Encoding and decoding strategies for wireless information transfer
with opportunistic WEH (via dynamic power splitting). The height of
block shown in the figure denotes the signal power.}\label{fig2}
\end{figure}

For the DPS scheme, we describe the enabling encoding and
decoding strategies for the following two
cases. Case I: $h(\nu)$ is unknown at Tx for all the
fading states of $\nu$, referred to as \emph{CSI Unknown at Tx}; and
Case II: $h(\nu)$ is perfectly known at Tx for each
fading state $\nu$, referred to as \emph{CSI Known at Tx} (CSIT).

First, consider the case of CSI Unknown at Tx, which is depicted in Fig. \ref{fig2}(a). In this case, Tx sends information continuously
with constant power $P$ for all the fading states due to the lack of
CSIT \cite{Shamai}. At each fading state $\nu$, Rx determines the optimal power ratio allocated to the information decoder $\alpha(\nu)$ and the energy harvester $1-\alpha(\nu)$, based on $h(\nu)$. For example, as shown in Fig. \ref{fig2}(a), in time slot 3, all the received power is allocated to the information decoder (i.e., $\alpha(\nu)=1$), while in time slots 1 and 2, the received power is split to both the information decoder and energy harvester (i.e., $0<\alpha(\nu)<1$).

Next, consider the case of CSIT as shown in Fig. \ref{fig2}(b). In this case, Tx is able to schedule
transmissions for information and energy transfer to Rx based on $h(\nu)$. As will be shown later in Section \ref{the case with CSIT}, the optimal power splitting rule in this case always has $\alpha(\nu)\neq 0$ provided that the transmitted power is non-zero. As a result, without loss of generality, we can assume that at any fading state $\nu$, Tx either transmits information signal or does not transmit at all (to save power). For example, in Fig. \ref{fig2}(b), Tx transmits information signal in time slots 1 and 3, and transmits no signal in time slot 2. Accordingly, Rx splits the received signal to the information decoder and the energy receiver (i.e., $0<\alpha(\nu)<1$) in slot 1, but allocates all the received power to the information receiver in time slot 3 (i.e., $\alpha(\nu)=1$). Moreover, Tx can implement power control based on
the instantaneous CSI to further improve the information and energy transmission
efficiency. Let $p(\nu)$ denote the transmit power of Tx at fading
state $\nu$. In this paper, we consider two types of power
constraints on $p(\nu)$, namely {\it average power constraint} (APC)
and {\it peak power constraint} (PPC). The APC
limits the average transmit power of Tx over all the fading states,
i.e., $E_{\nu}[p(\nu)]\leq P_{{\rm avg}}$, where $E_{\nu}[\cdot]$
denotes the expectation over $\nu$. In contrast, the PPC constrains
the instantaneous transmit power of Tx at each of the fading states,
i.e., $p(\nu)\leq P_{{\rm peak}}$, $\forall \nu$. Without loss of
generality, we assume $P_{{\rm avg}}\leq P_{{\rm peak}}$. For
convenience, we define the set of feasible power allocation as
\begin{align}\label{eqn:feasible power
set}\mathcal{P}\triangleq\big\{p(\nu): E_{\nu}[p(\nu)]\leq P_{{\rm
avg}}, p(\nu)\leq P_{{\rm peak}}, \forall \nu\big\}.\end{align}It is worth noting that for the case without CSIT, a fixed transmit power is assumed with $p(\nu)=P_{{\rm avg}}\triangleq P$, $\forall \nu$, such that both the APC and PPC are satisfied.

\section{Rate and Energy Trade-off in the SISO Fading Channel}\label{Ergodic Capacity and Harvested Energy Trade-off in Fading Channels}

In this paper, we consider the ergodic capacity as a relevant performance metric for information transfer. For the DPS scheme, given $\alpha(\nu)$ and $p(\nu)$, the instantaneous mutual information (IMI)
for the Tx-Rx link at fading state $\nu$ is expressed as\begin{align}\label{eqn:achievable rate}r(\nu)=\log\left(1+\frac{\alpha(\nu) h(\nu)p(\nu)}{\sigma^2}\right).\end{align}As a result, the ergodic capacity can be expressed as \cite{Shamai}\begin{align}R=E_\nu[r(\nu)].\end{align}
For information transfer, if
CSIT is not available, the ergodic capacity can be achieved by a
single Gaussian codebook with constant transmit power over all
different fading states \cite{Shamai99}; however, with CSIT, the
ergodic capacity can be further maximized by the ``water-filling (WF)''
based power allocation subject to the peak power constraint $P_{{\rm peak}}$ \cite{Goldsmith}, \cite{Khojastepour2004}.

On the other hand, for wireless energy transfer, the harvested energy (normalized by the slot duration) at each fading state $\nu$ can be expressed as $Q(\nu)=\xi(1-\alpha(\nu))h(\nu)p(\nu)$, where $\xi$ is a constant that
accounts for the loss in the energy transducer for converting the harvested energy to electrical energy to be stored; for convenience, it is assumed that $\xi=1$ in the rest of this paper unless stated otherwise. We thus have\begin{align}\label{eqn:harvested energy}Q(\nu)=(1-\alpha(\nu))h(\nu)p(\nu).\end{align}The average energy that is harvested at Rx is then given by\begin{align}Q_{{\rm avg}}=E_\nu[Q(\nu)].\end{align}

Evidently, there exist trade-offs in
assigning the power splitting ratio $\alpha(\nu)$ and/or transmit power
$p(\nu)$ (in the case of CSIT) to balance between maximizing the ergodic capacity for
information transfer versus maximizing the average harvested energy
for power transfer. To characterize such trade-offs, we adopt the so-called {\it Rate-Energy} (R-E) region (defined below) as introduced in \cite{Rui11}, \cite{Rui12}, which consists
of all the achievable ergodic capacity and average harvested energy
pairs given a power constraint $\mathcal{P}$ in (\ref{eqn:feasible power
set}). Specifically, in the case without (w/o) CSIT, the R-E region is defined
as\begin{align}\label{eqn:rate-energy region without
CSI}\mathcal{C}_{{\rm R-E}}^{{\rm w/o \ CSIT}}\triangleq &
\bigcup\limits_{p(\nu)=P,0\leq\alpha(\nu)\leq 1, \forall \nu} \bigg\{ (R,Q_{{\rm
avg}}):  \nonumber \\ & R \leq E_\nu[r(\nu)], Q_{{\rm avg}}\leq
E_\nu\left[Q(\nu)\right]\bigg\},\end{align}while in the case with CSIT, the R-E region is defined
as\begin{align}\label{eqn:rate-energy region with
CSI}\mathcal{C}_{{\rm R-E}}^{{\rm with \ CSIT}}\triangleq &
\bigcup\limits_{p(\nu)\in\mathcal{P}, 0\leq \alpha(\nu)\leq 1, \forall
\nu} \bigg\{(R,Q_{{\rm avg}}): \nonumber \\ & R \leq E_\nu[r(\nu)], Q_{{\rm
avg}}\leq E_\nu\left[Q(\nu)\right]\bigg\}.\end{align}

\begin{figure}
\begin{center}
\scalebox{0.56}{\includegraphics*[88,220][534,560]{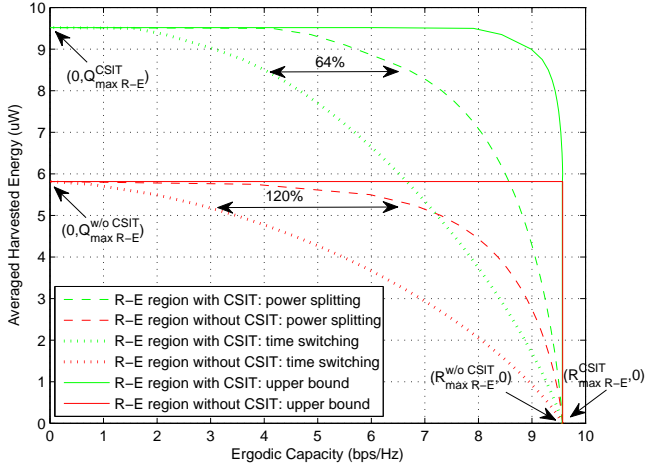}}
\end{center}
\caption{Examples of R-E region with versus without CSIT.}\label{fig3}
\end{figure}

Fig. \ref{fig3} shows some examples of the R-E region without versus with CSIT by the DPS scheme (see Section \ref{Rare-Energy Trade-off} for the details of computing these regions). It is assumed that the average transmit power constraint is $P_{{\rm avg}}=0.1$ watt(W) or $20$dBm, and the peak power constraint is $P_{{\rm peak}}=0.2$W or $23$dBm. The average operating distance between Tx and Rx is assumed to be $d=5$ meters, which results in an average of $40$dB signal power attenuation at a carrier frequency assumed as $f_c=900$MHz. With this distance, the line-of-sight (LOS) signal plays the dominant role, and thus Rician fading is used to model the channel. Specifically, at each fading state $\nu$, the complex channel can be modeled as $g(\nu)=\sqrt{\frac{K}{K+1}}\hat{g}+\sqrt{\frac{1}{K+1}}\tilde{g}(\nu)$, where $\hat{g}$ is the LOS deterministic component with $|\hat{g}|^2=-40$dB (to be consistent with the average path loss), $\tilde{g}(\nu)\sim \mathcal{CN}(0,-40{\rm dB})$ denotes the Rayleigh fading component, and $K$ is the Rician factor specifying the power ratio between the LOS and fading components in $g(\nu)$. Here we set $K=3$. The bandwidth of the transmitted signal is assumed to be $10$MHz, and the information receiver noise is assumed to be white Gaussian with power spectral density $-120$dBm/Hz or $-50$dBm over the entire bandwidth of $10$MHz. Moreover, the energy conversion efficiency for the energy harvester is assumed to be $\xi=0.5$. For comparison, we also show the R-E regions by a special form of DPS known as time switching \cite{Rui12} under the same channel setup with or without CSIT. Furthermore, the R-E regions obtained by assuming that the receiver can ideally decode information and harvest energy from the same received signal without any rate/power loss \cite{Sahai10} are added as a performance upper bound for DPS and time switching. It is observed that CSIT helps improve the achievable
R-E pairs at the receiver for both DPS and time switching schemes. Moreover, as compared to time switching, DPS achieves substantially improved R-E trade-offs towards the performance upper bound. For example, when $90\%$ of the maximum harvested energy is achieved, the ergodic capacity is increased by $64\%$ for the case with CSIT and $120\%$ for the case without CSIT, by comparing DPS versus time switching. It is also observed that when the average harvested power is smaller than $5.1$uW, DPS for the case without CSIT even outperforms time switching for the case with CSIT.


In Fig. \ref{fig3}, there are two boundary points shown in each R-E region, which are denoted by $(0,Q_{{\rm max}}^{{\rm w/o \ CSIT}})$, $(R_{{\rm max}}^{{\rm w/o \ CSIT}},0)$ for the case without CSIT, and $(0,Q_{{\rm max}}^{{\rm CSIT}})$, $(R_{{\rm max}}^{{\rm CSIT}},0)$ for the case with CSIT. For example, for the R-E trade-offs in the case without CSIT, we have
\begin{align}
& Q_{{\rm max}}^{{\rm w/o \ CSIT}}=E_\nu[h(\nu)P], \\ & R_{{\rm max}}^{{\rm w/o \ CSIT}}=E_\nu\left[\log\left(1+\frac{h(\nu)P}{\sigma^2}\right)\right].
\end{align}Note that $Q_{{\rm max}}^{{\rm w/o \ CSIT}}$ is achieved when $\alpha(\nu)=0$, $\forall \nu$, and thus the resulting ergodic capacity is zero, while $R_{{\rm max}}^{{\rm w/o \ CSIT}}$ is achieved when $\alpha(\nu)=1$, $\forall \nu$, and thus the resulting harvested energy is zero. The above holds for both time switching and power splitting receivers. Similarly, $Q_{{\rm max}}^{{\rm CSIT}}$ and $R_{{\rm max}}^{{\rm CSIT}}$ in the case with CSIT can be obtained, while for brevity, their
expressions are omitted here. It is worth noting that in general $R_{{\rm max}}^{{\rm CSIT}}>R_{{\rm max}}^{{\rm w/o \ CSIT}}$ due to the WF-based power control. However, with high signal-to-noise ratio (SNR), the rate gain by transmitter power control is negligibly small. As a result, in Fig. \ref{fig3} $R_{{\rm max}}^{{\rm CSIT}}$ and $R_{{\rm max}}^{{\rm w/o \ CSIT}}$ are observed to be very close to each other.

Since the optimal trade-offs between the ergodic capacity and the average harvested energy are
characterized by the boundary of the R-E region,
it is important to characterize all the boundary $(R,Q_{{\rm avg}})$ pairs for DPS in both the cases without and with
CSIT. Similarly as for the case of time switching in \cite{Rui12}, to characterize the Parato boundary of the R-E region for DPS, we need to solve the following two optimization problems.\begin{align*}\mathrm{(P1)}:~\mathop{\mathtt{Maximize}}_{\{\alpha(\nu)\}}
& ~~~ E_\nu[r(\nu)] \\
\mathtt {Subject \ to} & ~~~ E_\nu[Q(\nu)]\geq \bar{Q} \\ & ~~~
0\leq \alpha(\nu) \leq 1, \ \forall \nu
\end{align*}
\begin{align*}\mathrm{(P2)}:~\mathop{\mathtt{Maximize}}_{\{p(\nu),\alpha(\nu)\}}
& ~~~ E_\nu[r(\nu)] \\
\mathtt {Subject \ to} & ~~~ E_\nu[Q(\nu)] \geq \bar{Q} \\ & ~~~
p(\nu) \in \mathcal{P}, \ \forall \nu \\ & ~~~ 0\leq \alpha(\nu) \leq 1,
\ \forall \nu
\end{align*}where $\bar{Q}$ is a target average harvested energy required to maintain the receiver's operation. By solving Problem (P1) for all $0\leq \bar{Q}\leq Q_{{\rm max}}^{{\rm w/o \ CSIT}}$ and Problem (P2) for all $0\leq \bar{Q} \leq Q_{{\rm max}}^{{\rm CSIT}}$, we can characterize the entire boundary of the R-E region for the
case without CSIT (defined in (\ref{eqn:rate-energy region without
CSI})) and with CSIT (defined in (\ref{eqn:rate-energy region with
CSI})), respectively.


Problem (P1) is a convex optimization problem in terms of $\alpha(\nu)$'s, whereas Problem (P2) is non-convex in general since both the objective $E_\nu[r(\nu)]$ and harvested energy constraint $E_\nu[Q(\nu)]$ are non-concave functions over $\alpha(\nu)$ and $p(\nu)$. However, it can be verified that the Lagrangian duality method can still be
applied to solve Problem (P2) globally optimally, i.e., (P2) has strong duality or zero duality gap \cite{Boyd04}.

\begin{lemma}\label{lemma1}
Let $\{p^a(\nu),\alpha^a(\nu)\}$ and $\{p^b(\nu),\alpha^b(\nu)\}$ denote the
optimal solutions to Problem (P2) given the average harvested energy constraint and average transmit power constraint pairs $(\bar{Q}^a,P^a_{{\rm avg}})$
and $(\bar{Q}^b,P^b_{{\rm avg}})$, respectively. Then for any $0\leq \theta \leq 1$, there always exists a feasible solution $\{p^c(\nu),\alpha^c(\nu)\}$ such that \begin{align*}&E_\nu[r^c(\nu)]\geq \theta E_\nu[r^a(\nu)]+(1-\theta) E_\nu[r^b(\nu)], \\ &E_\nu[Q^c(\nu)]\geq \theta \bar{Q}^a+(1-\theta) \bar{Q}^b, \\ &E_\nu[p^c(\nu)]\leq \theta P_{{\rm avg}}^a+(1-\theta)P_{{\rm avg}}^b,\end{align*}where $r^\chi(\nu)=\log(1+\frac{h(\nu)\alpha^\chi(\nu)p^\chi(\nu)}{\sigma^2})$ with $\chi\in \{a,b,c\}$, and $Q^c(\nu)=(1-\alpha^c(\nu))h(\nu)p^c(\nu)$.
\end{lemma}
\begin{proof}
Please refer to Appendix \ref{appendix1}.
\end{proof}

Let $\Phi_2(\bar{Q},P_{{\rm avg}})$ denote the optimal value of (P2)
given the average harvested energy constraint $\bar{Q}$ and the average power constraint $P_{{\rm avg}}$. Lemma \ref{lemma1} implies that the ``time-sharing'' condition in \cite{Yu06} holds for (P2), and thus $\Phi_2(\bar{Q},P_{{\rm avg}})$ is concave in $(\bar{Q},P_{{\rm avg}})$, which then yields the zero duality gap of Problem (P2) according to the convex analysis in \cite{Boyd04}. Therefore, in the next section, we will apply the Lagrange duality method to
solve both (P1) and (P2).

\section{Optimal Policy for the SISO Fading Channel}\label{Rare-Energy Trade-off}

In this section, we study the optimal power splitting policy at Rx and/or
power control policy at Tx to achieve various optimal rate and energy trade-offs in the SISO fading channel
for both the cases without and with CSIT by solving Problems (P1) and
(P2), respectively.

\subsection{The Case Without CSIT}\label{the case without CSIT}
First, we consider Problem (P1) for the unknown CSIT case to determine the optimal power splitting rule at Rx with constant transmit power $P$ at Tx. The Lagrangian of Problem (P1) is expressed as\begin{align}L(\alpha(\nu),\lambda)=E_\nu[r(\nu)]+\lambda(E_\nu[Q(\nu)]-\bar{Q}),\end{align}where $\lambda\geq 0$ is the dual
variable associated with the harvested energy constraint $\bar{Q}$.
Then, the Lagrange dual function of Problem (P1) is given by
\begin{align}\label{eqn:dual function1}
g(\lambda)=\max\limits_{0\leq \alpha(\nu) \leq 1,\forall
\nu}L(\alpha(\nu),\lambda).
\end{align}

The maximization problem (\ref{eqn:dual
function1}) can be decoupled into parallel subproblems all having the
same structure and each for one fading state. For a particular
fading state $\nu$, the associated subproblem is expressed as
\begin{align}\label{eqn:subproblem3}\max_{0\leq \alpha \leq 1}  ~~~ L_{\nu}^{{\rm
w/o \ CSIT}}(\alpha),\end{align}where \begin{align}\label{eqn:1}L_{\nu}^{{\rm
w/o \ CSIT}}(\alpha)=r+\lambda Q=\log\left(1+\frac{\alpha hP}{\sigma^2}\right)+\lambda(1-\alpha)hP.\end{align}Note that in the above we have dropped the index $\nu$
for the fading state for brevity.

With a given $\lambda$, Problem (\ref{eqn:dual
function1}) can be efficiently solved by solving Problem (\ref{eqn:subproblem3}) for different fading states of $\nu$. Problem (P1) is then solved by iteratively solving (\ref{eqn:dual
function1}) with fixed $\lambda$, and updating $\lambda$ via a simple bisection method until the harvested energy constraint is met with equality \cite{Boyd04}. Let $\lambda^\ast$ denote the optimal dual solution that has a one-to-one correspondence to $\bar{Q}$ in Problem (P1). Then, we have the following proposition.

\begin{proposition}\label{proposition3}
The optimal solution to Problem (P1) is given by
\begin{align}\label{eqn:proposition3}\alpha^\ast(\nu)=\left\{\begin{array}{ll}\frac{1}{\lambda^\ast h(\nu)P}-\frac{\sigma^2}{h(\nu)P}, & {\rm if} \ h(\nu)\geq \frac{1}{\lambda^\ast P}-\frac{\sigma^2}{P}, \\ 1, & {\rm otherwise}.\end{array}\right.\end{align}

\end{proposition}

\begin{proof}
Please refer to Appendix \ref{appendix5}.
\end{proof}

It can be inferred from Proposition \ref{proposition3} that the power allocated to information decoding is a constant for all the fading states with $h(\nu)\geq \frac{1}{\lambda^\ast P}-\frac{\sigma^2}{P}$ since $\alpha^\ast(\nu) h(\nu)P=\frac{1}{\lambda^\ast}-\sigma^2\geq0$. Thus, $\lambda^\ast\leq \frac{1}{\sigma^2}$ must hold in (\ref{eqn:proposition3}). As a result, the achievable rate is a constant equal to $\log\frac{1}{\lambda^\ast \sigma^2}$ for such fading states. On the other hand, if $h(\nu)<\frac{1}{\lambda^\ast P}-\frac{\sigma^2}{P}$, all the received power is allocated to the information receiver. The above result is explained as follows. Suppose that if an amount of received power $\bar{P}$ is allocated to information receiver, we gain $\log(1+\frac{\bar{P}}{\sigma^2})$ in the achievable rate, but lose $\lambda^\ast \bar{P}$ in the harvested energy. Since the utility for our optimization problem given in (\ref{eqn:1}) at each fading state is the difference between the gain in the achievable rate and the loss in the harvested energy, the maximum utility is achieved when $\bar{P}^\ast=\frac{1}{\lambda^\ast}-\sigma^2$, which is a constant regardless of the fading state. Therefore, if the received power $h(\nu)P\geq \bar{P}^\ast$, i.e., $h(\nu)\geq \frac{1}{\lambda^\ast P}-\frac{\sigma^2}{P}$, then the received power allocated to the information receiver should be a constant $\frac{1}{\lambda^\ast}-\sigma^2$, and the remaining received power, i.e., $h(\nu)P-(\frac{1}{\lambda^\ast}-\sigma^2)$, is allocated to the energy receiver. Otherwise, if the received power is less than $\bar{P}^\ast$, it should be totally allocated to the information receiver.

In the following, we compare the above optimal receiver power splitting rule to the optimal time switching rule proposed in \cite{Rui12} for the achievable R-E trade-offs in the case without CSIT. For convenience, let $\lambda_{{\rm PS}}$ and $\lambda_{{\rm TS}}$ denote the optimal dual solutions to Problem (P1) with DPS and its modified problem (by changing the constraint $0\leq \alpha(\nu)\leq 1$ in (P1) to $\alpha(\nu)\in \{0,1\}$, $\forall \nu$) with time switching, respectively, for the same given $\bar{Q}$. Moreover, similarly as in \cite{Rui12}, we define the time switching indicator function as follows:\begin{align}\label{eqn:indicator
function}\alpha(\nu)=\left\{\begin{array}{ll}1, & {\rm ID \ mode \ is
\ active}
\\ 0, & {\rm EH \ mode \ is \ active}. \end{array} \right. \end{align}
We then have the following observations in order:
\begin{itemize}
\item When the fading state is ``poor'', i.e., $0<h(\nu)\leq \bar{h}$ for time switching or $0<h(\nu)\leq \frac{1}{\lambda_{{\rm PS}}P}-\frac{\sigma^2}{P}$ for power splitting, where $\bar{h}$ is the unique solution to the following equation given in \cite{Rui12}: $\log\left(1+\frac{hP}{\sigma^2}\right)=\lambda_{{\rm TS}}hP$, the optimal receiver strategy is to allocate all the received power to the information receiver for both cases of time switching and power splitting, i.e., $\alpha(\nu)=1$. In other words, the power allocated to information receiver is $h(\nu)P$, and that to energy receiver is $0$.
\item When the fading state is ``good'', i.e., $h(\nu)>\bar{h}$ for time switching or $h(\nu)>\frac{1}{\lambda_{{\rm PS}}P}-\frac{\sigma^2}{P}$ for power splitting, all the received power is allocated to energy harvester for time switching, i.e., $\alpha(\nu)=0$, while for power splitting, a constant power $\alpha(\nu)h(\nu)P=\frac{1}{\lambda_{{\rm PS}}}-\sigma^2$ is allocated to the information receiver, with the remaining power $(1-\alpha(\nu))h(\nu)P=h(\nu)P-\frac{1}{\lambda_{{\rm PS}}}+\sigma^2$ allocated to the energy receiver.
\end{itemize}
To summarize, the main difference between the optimal time switching and power splitting polices in the case without CSIT lies in the above ``good'' fading states. Specifically, both information decoding and energy harvesting can benefit from such good fading states if power splitting is used, while only energy harvesting benefits if time switching is used. An illustration of the above difference in the received power allocation for power splitting versus time switching is given in Fig. \ref{fig4}.

\begin{figure}
\begin{center}
\scalebox{0.5}{\includegraphics*{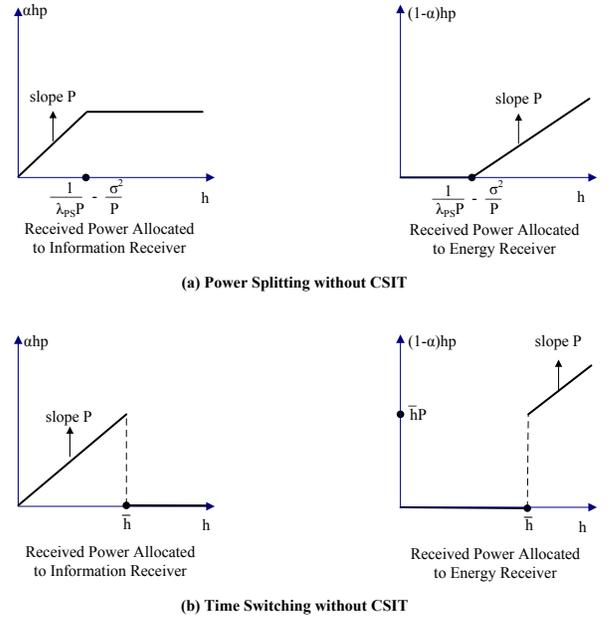}}
\end{center}
\caption{Power splitting versus time switching: a comparison of the received power allocation to information receiver and energy receiver over different fading states for the case without CSIT.}\label{fig4}
\end{figure}

\subsection{The Case With CSIT}\label{the case with CSIT}
For the case with CSIT, in addition to the receiver's DPS, the transmitter can implement power control to further improve the R-E trade-off. To jointly optimize the values of $p(\nu)$ and $\alpha(\nu)$, $\forall \nu$, we need to solve Problem (P2), shown as follows.

Let $\lambda$ and $\beta$ denote the nonnegative
dual variables corresponding to the average harvested energy
constraint and average transmit power constraint in Problem (P2), respectively.
Similarly as for Problem (P1), Problem (P2) can be decoupled into
parallel subproblems each for one particular fading state and
expressed as (by ignoring the fading index
$\nu$)\begin{align}\label{eqn:subproblem4}\max_{0\leq p\leq P_{{\rm
peak}},0\leq \alpha \leq 1}  ~~~ L_{\nu}^{{\rm
with \ CSIT}}(p,\alpha),\end{align}where \begin{align}\label{eqn:2}L_{\nu}^{{\rm
with \ CSIT}}(p,\alpha)=&r+\lambda Q-\beta p \nonumber \\ =& \log\left(1+\frac{\alpha hp}{\sigma^2}\right)+\lambda (1-\alpha)hp-\beta p.\end{align}

After solving Problem (\ref{eqn:subproblem4}) with given $\lambda$ and $\beta$ for all the fading states, we can update $(\lambda,\beta)$ via the ellipsoid method \cite{Boyd04}. It can be shown
that the sub-gradient for updating $(\lambda,\beta)$ is
$(E_\nu[Q^\ast(\nu)]-\bar{Q},P_{{\rm avg}}-E_\nu[p^\ast(\nu)])$,
where $Q^\ast(\nu)$ and $p^\ast(\nu)$ denote the harvested energy
and transmit power at fading state $\nu$, respectively, obtained by
solving Problem (\ref{eqn:subproblem4}) for a given pair of
$\lambda$ and $\beta$. Let $\lambda^\ast$ and $\beta^\ast$ denote the optimal dual solutions to Problem (P2) for a given set of $\bar{Q}$, $P_{{\rm avg}}$ and $P_{{\rm peak}}$. Similarly as for Proposition \ref{proposition3}, it can be shown that when $0\leq \bar{Q} \leq Q_{{\rm max}}^{{\rm CSIT}}$, it must hold that $\lambda^\ast<\frac{1}{\sigma^2}$. We then have the following proposition.

\begin{proposition}\label{proposition4}
By defining $\tilde{h}=\frac{1}{\lambda^\ast P_{{\rm peak}}}-\frac{\sigma^2}{P_{{\rm peak}}}$, the optimal solution to Problem (P2) is given by

If $\frac{\beta^\ast}{\lambda^\ast}\leq \tilde{h}$,
\small{\begin{align}\left\{\begin{array}{ll}p^\ast(\nu)=P_{{\rm peak}}, \ \alpha^\ast(\nu)=\frac{\tilde{h}}{h(\nu)}, & {\rm if} \ h(\nu)\geq \tilde{h}, \\ p^\ast(\nu)=P_{{\rm peak}}, \ \alpha^\ast(\nu)=1, & \frac{\beta^\ast\sigma^2}{1-\beta^\ast P_{{\rm peak}}}\leq h(\nu)< \tilde{h}, \\ p^\ast(\nu)=\frac{1}{\beta^\ast}-\frac{\sigma^2}{h(\nu)}, \ \alpha^\ast(\nu)=1, & \beta^\ast \sigma^2\leq h(\nu) <\frac{\beta^\ast\sigma^2}{1-\beta^\ast P_{{\rm peak}}}, \\ p^\ast(\nu)=0, & {\rm otherwise}\end{array}\right.\end{align}}

if $\frac{\beta^\ast}{\lambda^\ast}>\tilde{h}$,\begin{align}\left\{\begin{array}{ll}p^\ast(\nu)=P_{{\rm peak}}, \ \alpha^\ast(\nu)=\frac{\tilde{h}}{h(\nu)}, & {\rm if} \ h(\nu)\geq \frac{\beta^\ast}{\lambda^\ast}, \\ p^\ast(\nu)=\frac{1}{\beta^\ast}-\frac{\sigma^2}{h(\nu)}, \ \alpha^\ast(\nu)=1, & \beta^\ast \sigma^2\leq h(\nu) <\frac{\beta^\ast}{\lambda^\ast}, \\ p^\ast(\nu)=0, & {\rm otherwise}.\end{array}\right.\end{align}
\end{proposition}

\begin{proof}
Please refer to Appendix \ref{appendix6}.
\end{proof}

It is worth noting that similar to the case without CSIT, from Proposition \ref{proposition4} it follows that in the case with CSIT, if $h(\nu)>\max(\frac{\beta^\ast}{\lambda^\ast},\tilde{h})$, a constant received power $\alpha(\nu)^\ast h(\nu)p^\ast(\nu)=\frac{1}{\lambda^\ast}-\sigma^2$ is allocated to the information receiver, while the remaining received power is allocated to the energy receiver; otherwise, if $h(\nu)\leq \max(\frac{\beta^\ast}{\lambda^\ast},\tilde{h})$, all the received power is allocated to the information receiver.


Next, we compare the optimal power splitting and time switching for the achievable R-E trade-offs in the case with CSIT. For convenience, let $(\lambda_{{\rm PS}},\beta_{{\rm PS}})$ and $(\lambda_{{\rm TS}},\beta_{{\rm TS}})$ denote the optimal dual solutions to Problem (P2) with DPS and its modified form (by changing the constraint $0\leq \alpha(\nu)\leq 1$ in (P2) as $\alpha(\nu)\in \{0,1\}$, $\forall \nu$) with time switching, respectively. We then obtain the following observations:
\begin{itemize}
\item When the fading state is ``poor'', i.e., $0<h(\nu)\leq \beta_{{\rm TS}}\sigma^2$ for time switching or $0<h(\nu)\leq \beta_{{\rm PS}}\sigma^2$ for power splitting, the optimal strategy is to switch off the transmission to save transmit power in both schemes.
\item For moderate fading states with $\beta_{{\rm TS}}\sigma^2< h(\nu)\leq \hat{h}$ for time switching or $\beta_{{\rm PS}}\sigma^2<h(\nu)\leq \max(\frac{\beta_{\rm PS}}{\lambda_{\rm PS}},\tilde{h})$ for power splitting, where $\hat{h}$ is the largest root of the following equation given in \cite{Rui12}: $\log\frac{h}{\beta_{{\rm TS}}\sigma^2}-1+\frac{\beta_{{\rm TS}}\sigma^2}{h}-\lambda_{{\rm TS}} hP_{{\rm peak}}+\beta_{{\rm TS}} P_{{\rm peak}}=0$, the optimal strategy is to transmit information with water-filling power allocation at Tx (with the maximum transmit power capped by $P_{{\rm peak}}$) and allocate all the received power to information receiver in both schemes.

\item When the fading state is ``good'', i.e., $h(\nu)>\hat{h}$ for time switching or $h(\nu)> \max(\frac{\beta_{\rm PS}}{\lambda_{\rm PS}},\tilde{h})$ for power splitting, the optimal strategy of the transmitter is to transmit at peak power $P_{{\rm peak}}$ in both schemes. However, at the receiver, all the received power is allocated to the energy receiver for time switching, i.e., $(1-\alpha(\nu))h(\nu)p(\nu)=h(\nu)P_{{\rm peak}}$, while for power splitting, only a constant amount of the received power $\alpha(\nu)h(\nu)p(\nu)=\frac{1}{\lambda_{{\rm PS}}}-\sigma^2$ is allocated to information receiver with the remaining power $(1-\alpha(\nu))h(\nu)p(\nu)=h(\nu)P_{{\rm peak}}-\frac{1}{\lambda_{{\rm PS}}}+\sigma^2$ allocated to the energy receiver.
\end{itemize}
To summarize, similar to the case without CSIT, the main difference between the optimal resource allocation polices between power splitting and time switching for the case with CSIT lies in the above ``good'' fading states. Specifically, both information decoding and energy harvesting can benefit from good fading states if power splitting is used, while only energy harvesting benefits if time switching is used. An illustration of the above transmitter power control and receiver power allocation policies for power splitting versus time switching is given in Fig. \ref{fig5}.

\begin{figure}
\begin{center}
\subfigure[power splitting with CSIT: $\frac{\beta_{{\rm PS}}}{\lambda_{{\rm PS}}}\geq \frac{1}{\lambda_{{\rm PS}}P_{{\rm peak}}}-\frac{\sigma^2}{P_{{\rm peak}}}$]{\scalebox{0.45}{\includegraphics*{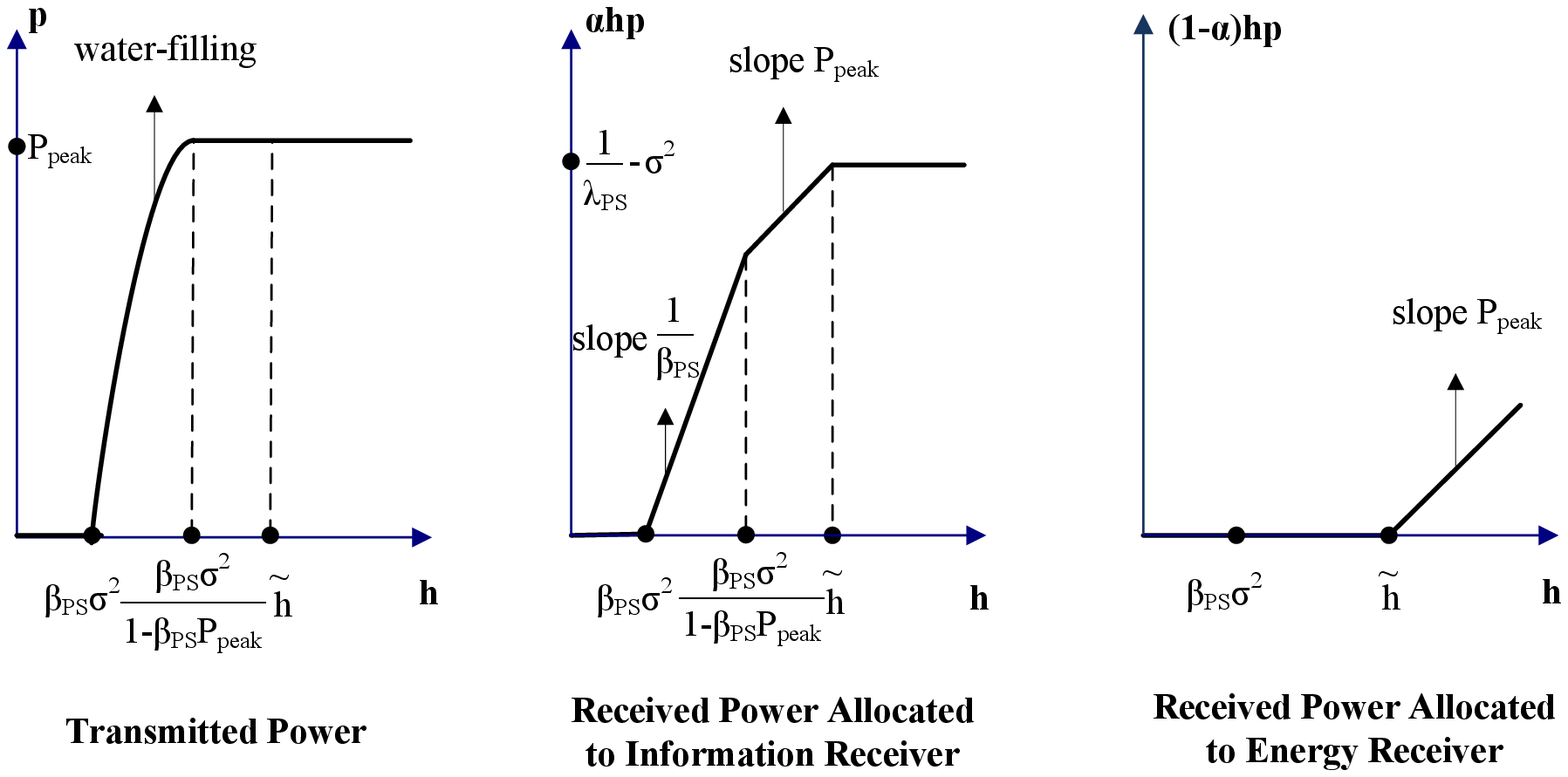}}} \\
\subfigure[power splitting with CSIT: $\frac{\beta_{{\rm PS}}}{\lambda_{{\rm PS}}}< \frac{1}{\lambda_{{\rm PS}}P_{{\rm peak}}}-\frac{\sigma^2}{P_{{\rm peak}}}$]{\scalebox{0.45}{\includegraphics*{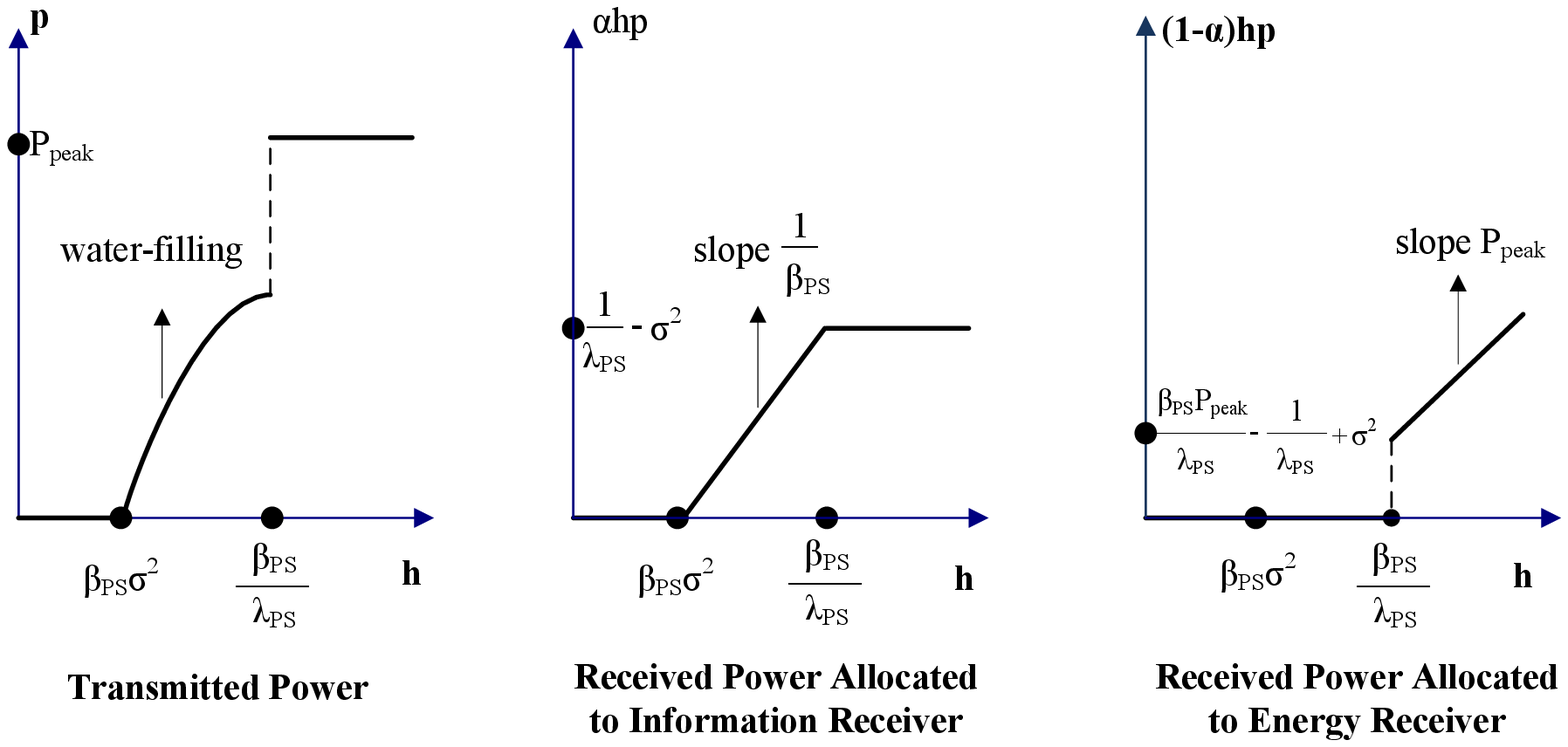}}} \\
\subfigure[time switching with CSIT]{\scalebox{0.45}{\includegraphics*{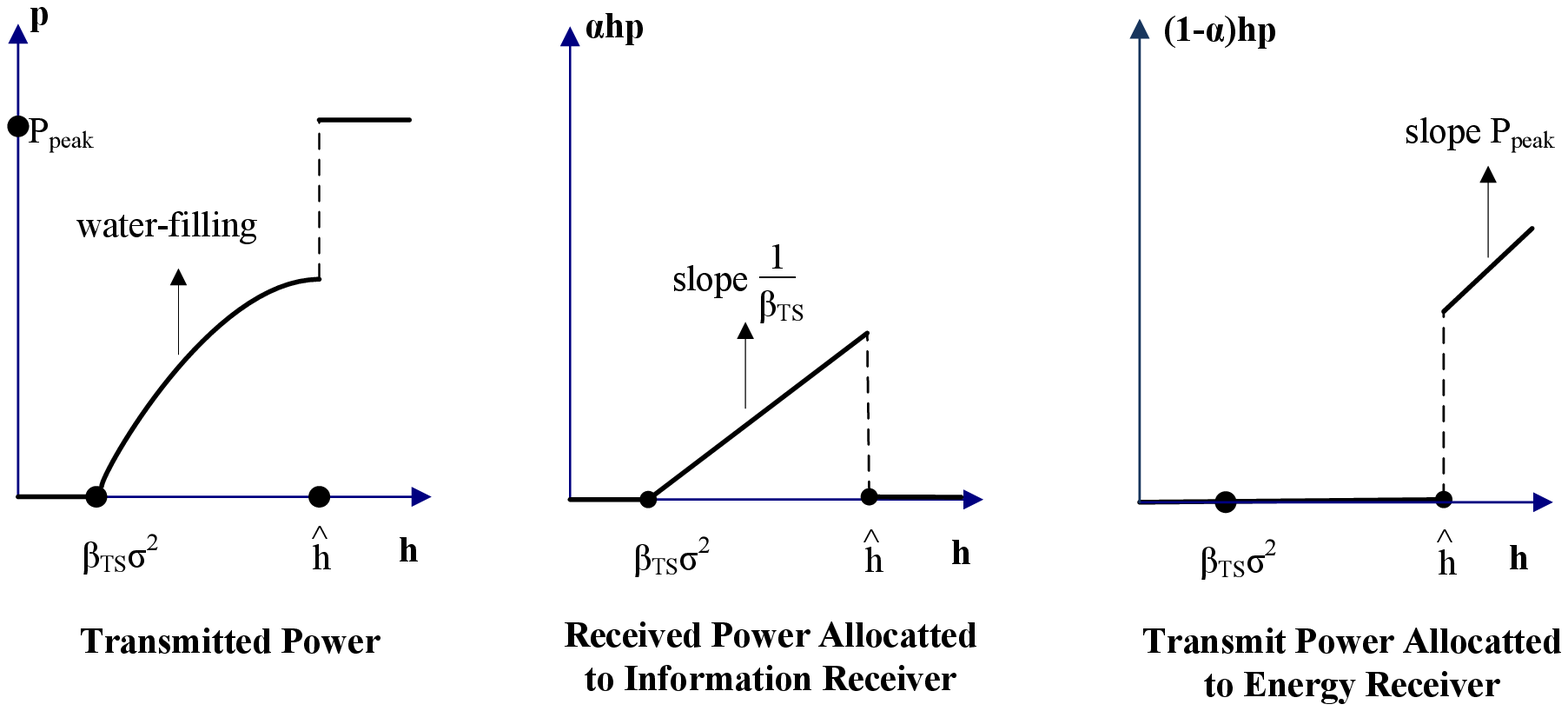}}}
\end{center}
\caption{Power splitting versus time switching: a comparison of the transmit power allocation and received power allocation over different fading states for the case with CSIT.}\label{fig5}
\end{figure}

\subsection{Performance Upper Bound}\label{sec:The Upper Bound of the R-E Trade-offs}
In this subsection, we derive a R-E region upper bound for DPS (as well as other practical receiver designs) by considering an ideal receiver that can simultaneously decode information and harvest energy from the same received signal without any information/energy loss. This is equivalent to setting $\alpha(\nu)=1$ $\forall \nu$ in (\ref{eqn:achievable rate}) and $\alpha(\nu)=0$ $\forall \nu$ in (\ref{eqn:harvested energy}) at the same time. In this case, the information rate and harvested energy at each fading state $\nu$ can be respectively expressed as
\begin{align}
& r(\nu)=\log\left(1+\frac{h(\nu)p(\nu)}{\sigma^2}\right), \label{eqn:upper rate}\\
& Q(\nu)=h(\nu)p(\nu). \label{eqn:upper harvested energy}
\end{align}

For the case without CSIT, there is no trade-off between information and energy transfer from the above since $p(\nu)=P$, $\forall \nu$. As a result, as shown in Fig. \ref{fig3}, the R-E region upper bound for the case without CSIT is simply a box. On the other hand, in the case with CSIT, a trade-off between $r(\nu)$ and $Q(\nu)$ given in (\ref{eqn:upper rate}) and (\ref{eqn:upper harvested energy}) due to the power allocation policy $p(\nu)$, which has been similarly studied in \cite{Sahai10} for the frequency-selective AWGN channel with simultaneous information and power transfer. By solving Problem (P2) for all feasible $\bar{Q}$'s, with $r(\nu)$ and $Q(\nu)$ replaced by (\ref{eqn:upper rate}) and (\ref{eqn:upper harvested energy}), respectively, the R-E region upper bound in the case of CSIT can be obtained. Let $\lambda^\ast$ and $\beta^\ast$ denote the optimal dual solutions associated with the harvested energy constraint $\bar{Q}$ and average power constraint $P_{{\rm avg}}$, respectively. By following the similar proof of Proposition \ref{proposition4}, we can obtain the optimal power allocation for achieving the R-E region upper bound in the case with CSIT in the following proposition.

\begin{proposition}\label{proposition5}
For Problem (P2) with $r(\nu)$ and $Q(\nu)$ replaced by (\ref{eqn:upper rate}) and (\ref{eqn:upper harvested energy}), respectively, and the constraint $0\leq \alpha(\nu)\leq 1$ being removed, the optimal power allocation is given by
\begin{align}
p(\nu)=\left\{\begin{array}{ll} P_{{\rm peak}}, & {\rm if} \ h(\nu)\geq \frac{\beta^\ast}{\lambda^\ast}, \\ \left[\frac{1}{\beta^\ast-\lambda^\ast h(\nu)}-\frac{\sigma^2}{h(\nu)}\right]_0^{P_{{\rm peak}}}, & {\rm otherwise}, \end{array}\right.
\end{align}where $[x]_a^b=\max(\min(x,b),a)$.
\end{proposition}

\section{Extension and Application: Dynamic Power Splitting for the SIMO Fading Channel}\label{sec:Implementation of Power Splitting by Multiple Receive Antennas}

%


In this section, we extend the result for DPS to the SIMO fading channel, i.e., when the receiver is equipped with multiple antennas, and furthermore study a low-complexity implementation of power splitting, namely antenna switching \cite{Rui11}.

\subsection{Optimal Power Splitting}\label{sec:Optimal Power Splitting in SIMO System}

\begin{figure}
\begin{center}
\scalebox{0.45}{\includegraphics*{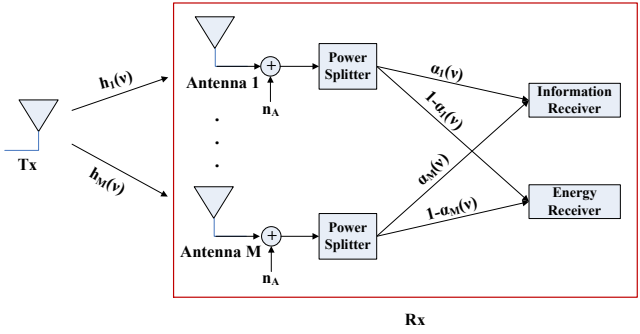}}
\end{center}
\caption{DPS for the SIMO system.}\label{fig9}
\end{figure}

First, we study the optimal DPS scheme for the SIMO system, as shown in Fig. \ref{fig9}. Assuming that the receiver is equipped with $M>1$ antennas, then at any fading state $\nu$, the complex channel and the channel power gain from Tx to the $m$th antenna of Rx are denoted by $g_m(\nu)$ and $h_m(\nu)=|g_m(\nu)|^2$, $1\leq m \leq M$, respectively. Without loss of generality, similar to the SISO case, at fading state $\nu$, each receiving antenna $m$ can split $0\leq \alpha_m(\nu)\leq 1$ portion of the received signal power to the information receiver, and the remaining $1-\alpha_m(\nu)$ portion of power to the energy receiver.

For the information receiver, it is assumed that the maximal ratio combining (MRC) is applied over the signals split from the $M$ receiving antennas. Therefore, at fading state $\nu$, the achievable rate can be expressed as\begin{align}
r(\nu)=\log\left(1+\sum\limits_{m=1}^M\frac{\alpha_m(\nu) h_m(\nu)p(\nu)}{\sigma^2}\right). \label{eqn:new rate}
\end{align}Moreover, the total harvested energy from the signals split from the $M$ receiving antennas at the energy receiver can be expressed as\begin{align}Q(\nu)=\sum\limits_{m=1}^M(1-\alpha_m(\nu))h_m(\nu)p(\nu). \label{eqn:new harvested energy}
\end{align}Then, with $r(\nu)$ and $Q(\nu)$ given by (\ref{eqn:new rate}) and (\ref{eqn:new harvested energy}), we can define the achievable R-E regions for the SIMO system in both the cases without and with CSIT as $\mathcal{C}_{{\rm R-E}}^{{\rm w/o \ CSIT \ (SIMO)}}$ and $\mathcal{C}_{{\rm R-E}}^{{\rm CSIT \ (SIMO)}}$, respectively, similarly to (\ref{eqn:rate-energy region without CSI}) and (\ref{eqn:rate-energy region with CSI}) in the SISO case, and characterize their boundaries by solving problems similarly to (P1) and (P2).

\subsubsection{The Case Without CSIT}

\ \ \ \ \

First, we study the optimal DPS for the case without CSIT in the SIMO fading channel to obtain $\mathcal{C}_{{\rm R-E}}^{{\rm w/o \ CSIT \ (SIMO)}}$. Given $p(\nu)=P$, $\forall \nu$, similar to solving (P1) in Section \ref{the case without CSIT}, by introducing the Lagrange dual variable $\lambda$ associated with the energy constraint $\bar{Q}$, the optimization problem for the SIMO system can be decoupled into parallel subproblems each for one fading state, which is expressed as (by ignoring the fading index $\nu$)\begin{align}\label{eqn:subproblem5}\max_{\{0\leq \alpha_m \leq 1\}}  ~~~  L_{\nu}^{{\rm
w/o \ CSIT \ (SIMO)}}(\{\alpha_m\}),\end{align}where \begin{align}\label{eqn:5} & L_{\nu}^{{\rm
w/o \ CSIT \ (SIMO)}}(\{\alpha_m\}) \nonumber \\=& r+\lambda Q \nonumber \\ =& \log\left(1+\frac{\sum\limits_{m=1}^M\alpha_m h_mP}{\sigma^2}\right)+\sum\limits_{m=1}^M\lambda(1-\alpha_m)h_mP.\end{align}

\begin{lemma}\label{lem1}
Given any fixed $\lambda$, Problem (\ref{eqn:subproblem5}) is equivalent to the following problem:
\begin{align}\mathop{\mathtt{Maximize}}_{\alpha}
& ~~~ \log\left(1+\frac{\alpha\sum\limits_{m=1}^M h_mP}{\sigma^2}\right)+(1-\alpha)\sum\limits_{m=1}^M\lambda h_mP \nonumber \\
\mathtt {Subject \ to} & ~~~ 0\leq \alpha \leq 1. \label{eqn:subproblem6}
\end{align}
\end{lemma}

\begin{proof}
Given any $\alpha_m\in [0,1]$, $1\leq m \leq M$, it follows that $\alpha\triangleq \frac{\sum\limits_{m=1}^M\alpha_m h_m P}{\sum\limits_{m=1}^Mh_m P}$ lies in $[0,1]$ and achieves the same objective value of Problem (\ref{eqn:subproblem6}) as that of Problem (\ref{eqn:subproblem5}). Thus, the optimal value of Problem (\ref{eqn:subproblem6}) must be no smaller than that of Problem (\ref{eqn:subproblem5}). On the other hand, given any $0\leq \alpha \leq 1$, there exists at least one solution for $\alpha_m$'s such that $\sum\limits_{m=1}^M\alpha_m h_m P=\alpha\sum\limits_{m=1}^Mh_m P$ with $0\leq \alpha_m \leq 1$, $\forall m$. Thus, the optimal value of Problem (\ref{eqn:subproblem5}) must be no smaller than that of Problem (\ref{eqn:subproblem6}). Therefore, Problems (\ref{eqn:subproblem5}) and (\ref{eqn:subproblem6}) have the same optimal value and thus are equivalent. Lemma \ref{lem1} is thus proved.
\end{proof}

Lemma \ref{lem1} suggests that a ``uniform power splitting (UPS)'' scheme by setting $\alpha_m=\alpha$, $\forall m$, is in fact optimal to achieve the boundary of $\mathcal{C}_{{\rm R-E}}^{{\rm w/o \ CSIT \ (SIMO)}}$ in the SIMO fading channel without CSIT. More interestingly, Lemma \ref{lem1} establishes the equivalence between the optimal DPS policies for the SIMO and SISO systems, given as follows. By comparing Problem (\ref{eqn:subproblem3}) in the SISO case and Problem (\ref{eqn:subproblem6}) in the SIMO case, it is observed that if $h$ is replaced by $\sum\limits_{m=1}^Mh_m$, then Problem (\ref{eqn:subproblem3}) is the same as Problem (\ref{eqn:subproblem6}). Therefore, in the SIMO case, we can treat all the receiving antennas as one ``virtual'' antenna with an equivalent channel sum-power gain from Tx as $h=\sum\limits_{m=1}^Mh_m$; thereby, the SIMO system in Fig. \ref{fig9} becomes equivalent to a SISO system that has been studied in Section \ref{the case without CSIT}. Hence, by replacing $h(\nu)$ by $\sum\limits_{m=1}^Mh_m(\nu)$ and letting $\alpha_m(\nu)=\alpha(\nu)$, $\forall m,\nu$, the optimal UPS solution for the SIMO fading channel can similarly be obtained by Proposition \ref{proposition3} in the SISO case, for which the details are omitted for brevity.


\subsubsection{The Case With CSIT}

\ \ \ \ \

Next, we consider the joint DPS at Rx and power control at Tx for the SIMO case with CSIT to characterize the boundary of $\mathcal{C}_{{\rm R-E}}^{{\rm CSIT \ (SIMO)}}$. Similar to solving Problem (P2) in Section \ref{the case with CSIT}, by introducing the Lagrange dual variables $\lambda$ and $\beta$ associated with the energy constraint $\bar{Q}$ and average power constraint $P_{{\rm avg}}$, respectively, the optimization problem for the SIMO fading channel can be decoupled into parallel subproblems each for one fading state, which is expressed as (by ignoring the fading index $\nu$)\begin{align}\label{eqn:subproblem7}\max_{0\leq p\leq P_{{\rm peak}},\{0\leq \alpha_m \leq 1\}}  ~~~ L_{\nu}^{{\rm
with \ CSIT \ (SIMO)}}(\{\alpha_m\},p),\end{align}where \begin{align}\label{eqn:7} & L_{\nu}^{{\rm
with \ CSIT \ (SIMO)}}(\{\alpha_m\},p) \nonumber \\ =&r+\lambda Q \nonumber \\ =& \log\left(1+\frac{\sum\limits_{m=1}^M\alpha_m h_mp}{\sigma^2}\right)+\sum\limits_{m=1}^M\lambda(1-\alpha_m)h_mp-\beta p.\end{align}

Similar to Lemma \ref{lem1}, the following lemma establishes the optimality of UPS in the SIMO case with CSIT.

\begin{lemma}\label{lem2}
Given any fixed $\lambda$ and $\beta$, Problem (\ref{eqn:subproblem7}) is equivalent to the following problem:
\small{\begin{align}\mathop{\mathtt{Maximize}}_{p,\alpha}
& ~~~ \log\left(1+\frac{\alpha\sum\limits_{m=1}^M h_m p}{\sigma^2}\right)+(1-\alpha)\sum\limits_{m=1}^M\lambda h_m p-\beta p \nonumber \\
\mathtt {Subject \ to} & ~~~ 0\leq \alpha \leq 1, \nonumber \\ & ~~~ 0\leq p \leq P_{{\rm peak}}. \label{eqn:subproblem8}
\end{align}}
\end{lemma}

The proof of Lemma \ref{lem2} is similar to that of Lemma \ref{lem1}, and is thus omitted for brevity. Lemma \ref{lem2} implies that the equivalence between the SIMO and SISO systems also holds in the case with CSIT, by treating all the receiving antennas in the SIMO system as one ``virtual'' antenna in the SISO system with the equivalent channel power gain given by $h=\sum\limits_{m=1}^Mh_m$. As for the case without CSIT, the optimal transmitter power allocation $p(\nu)$ and receiver UPS $\alpha_m(\nu)=\alpha(\nu)$, $\forall m,\nu$, for the SIMO fading channel with CSIT can similarly be obtained from Proposition \ref{proposition4} in the SISO case by replacing $h(\nu)$ by $\sum\limits_{m=1}^Mh_m(\nu)$.

\subsection{Antenna Switching}\label{sec:Antenna Switching to Approach Optimal Power Splitting}

Note that the optimal UPS for the SIMO system requires multiple power splitters each equipped with one receiving antenna to adjust the power splitting ratio at each fading state. Practically, this could be very costly to implement. Therefore, in this subsection we consider a low-complexity implementation for power splitting in the SIMO system with multiple receiving antennas, namely antenna switching \cite{Rui11}. As shown in Fig. \ref{fig12}, at each fading state, instead of splitting the power at each receiving antenna, the antenna switching scheme simply connects one subset of the receiving antennas (denoted by $\Phi_{{\rm ID}}(\nu)$) to information receiver, with the remaining subset of antennas (denoted by $\Phi_{{\rm EH}}(\nu)$) to energy harvester, i.e.,\begin{align}\label{eqn:indicator
function}\alpha_m(\nu)=\left\{\begin{array}{ll}1, & {\rm if} \ m\in \Phi_{{\rm ID}},
\\ 0, & {\rm if} \ m\in \Phi_{{\rm EH}}, \end{array} \right.  1\leq m \leq M.\end{align}It is worth noting that antenna switching can be shown equivalent to UPS with $\alpha_m(\nu)=\frac{\sum\limits_{m\in \Phi_{{\rm ID}}}h_m(\nu)p(\nu)}{\sum\limits_{m=1}^Mh_m(\nu)p(\nu)}$ for $\forall m, \nu$. However, since antenna switching only requires the time switcher at each receiving antenna instead of the more costly power splitter in UPS, it is practically more favorable. In the following, we study the optimal antenna switching policy for the SIMO fading channel without or with CSIT.

\begin{figure}
\begin{center}
\scalebox{0.45}{\includegraphics*{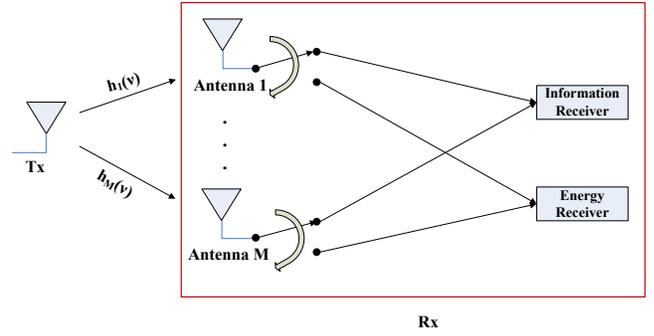}}
\end{center}
\caption{Antenna switching for the SIMO system.}\label{fig12}
\end{figure}

\subsubsection{The Case Without CSIT}

\ \ \ \ \

In this case, the optimal antenna switching rule can be obtained by solving Problem (\ref{eqn:subproblem5}) with $\alpha_m(\nu)$'s given in (\ref{eqn:indicator function}). However, to find the optimal antenna partitions, i.e., $\Phi_{{\rm ID}}^\ast(\nu)$ and $\Phi_{{\rm EH}}^\ast(\nu)$, at any fading state $\nu$, we need to search over $2^M$ possible antenna combinations to maximize (\ref{eqn:5}), for which the complexity goes up exponentially as $M$ increases.

\subsubsection{The Case With CSIT}

\ \ \ \ \

In this case, the optimal transmitter power control $p^\ast(\nu)$ and receiver antenna switching $\Phi_{{\rm ID}}^\ast(\nu)$ and $\Phi_{{\rm EH}}^\ast(\nu)$ at each fading state $\nu$ can be obtained by solving Problem (\ref{eqn:subproblem7}) with $\alpha_m$'s given by (\ref{eqn:indicator function}). First, given any policy of $\Phi_{{\rm ID}}(\nu)$ and $\Phi_{{\rm EH}}(\nu)$, Problem (\ref{eqn:subproblem7}) reduces to the following problem:
\begin{align}\mathop{\mathtt{Maximize}}_{p}
& ~~~ \log\left(1+\frac{\sum\limits_{m \in \Phi_{{\rm ID}}} h_m p}{\sigma^2}\right)+\sum\limits_{m\in \Phi_{{\rm EH}}}\lambda h_m p-\beta p \nonumber \\ \mathtt {Subject \ to} & ~~~ 0\leq p \leq P_{{\rm peak}}. \label{eqn:subproblem9}
\end{align}It can be shown that the optimal solution to Problem (\ref{eqn:subproblem9}) is given by
\begin{small}
\begin{align}\label{eqn:antenna switching power}
p^\ast=\left\{\begin{array}{ll}P_{{\rm peak}}, & {\rm if} \ \sum\limits_{m\in \Phi_{{\rm EH}}}h_m\geq \frac{\beta}{\lambda}, \\ \left[\frac{1}{\beta-\lambda \sum\limits_{m\in \Phi_{{\rm EH}}}h_m}-\frac{\sigma^2}{\sum\limits_{m\in \Phi_{{\rm ID}}}h_m}\right]_0^{P_{{\rm peak}}} . & {\rm otherwise} \end{array} \right.
\end{align}\end{small}Therefore, given any antenna partitions $\Phi_{{\rm ID}}(\nu)$ and $\Phi_{{\rm EH}}(\nu)$, the value of (\ref{eqn:7}) can be obtained by (\ref{eqn:antenna switching power}). Then, the optimal $\Phi_{{\rm ID}}^\ast(\nu)$ and $\Phi_{{\rm EH}}^\ast(\nu)$ can be found by searching over all $2^M$ possible antenna combinations to maximize the resulting value of (\ref{eqn:7}).

\subsection{Low-Complexity Antenna Switching Algorithm}\label{sec:A Heuristic Antenna Switching Scheme with Lower Complexity}

Although antenna switching reduces the hardware complexity as compared to power splitting for the SIMO system, its optimal policy by the exhaustive search as shown in the previous subsection is of exponentially increasing complexity with the number of receiving antennas $M$. In this subsection, we propose a low-complexity algorithm for antenna switching which only has a polynomial complexity in the order of $\mathcal{O}(M^2)$ instead of $\mathcal{O}(2^M)$ by the exhaustive search. Instead of solving Problems (\ref{eqn:subproblem5}) and (\ref{eqn:subproblem7}) directly with $\alpha_m(\nu)$'s given in (\ref{eqn:indicator function}), the proposed algorithm first solves the optimal UPS policy (see Section \ref{sec:Optimal Power Splitting in SIMO System}) by treating the SIMO system as an equivalent SISO system with one virtual antenna and then efficiently finds a pair of $\Phi_{{\rm ID}}(\nu)$ and $\Phi_{{\rm EH}}(\nu)$ to approximate the obtained UPS solution as close as possible.

\subsubsection{The Case Without CSIT}

\ \ \ \ \

From Proposition \ref{proposition3}, it is known that the optimal UPS policy for the equivalent SISO system (with channel power gain $h(\nu)=\sum\limits_{m=1}^Mh_m(\nu)$) in the case of SIMO system without CSIT allocates $\frac{1}{\lambda^\ast}-\sigma^2$ amount of power to the information receiver if the total received power $\sum\limits_{m=1}^Mh_m(\nu)P$ is larger than $\frac{1}{\lambda^\ast}-\sigma^2$; otherwise, all the received power is allocated to the information receiver (c.f. Fig. \ref{fig4}(a)). Therefore, to approximate the optimal UPS policy in the case without CSIT, at each fading state we should find a solution for antenna switching such that $\sum\limits_{m\in \Phi_{{\rm ID}}(\nu)}h_m(\nu)P$ is as close to $\frac{1}{\lambda^\ast}-\sigma^2$ as possible. On the other hand, to satisfy the average harvested energy constraint $\bar{Q}$, $\sum\limits_{m\in \Phi_{{\rm ID}}(\nu)}h_m(\nu)P$ should be no large than $\frac{1}{\lambda^\ast}-\sigma^2$. Hence, by defining the set $\mathcal{T}=\{h_1(\nu)P,\cdots,h_M(\nu)P\}$, our proposed antenna switching algorithm searches for a subset of $\mathcal{T}$ that has the sum of elements closest to, but no larger than $\frac{1}{\lambda^\ast}-\sigma^2$. This leads to the following problem at each fading state of $\nu$.
\begin{align*}\mathrm{(P3)}:~\mathop{\mathtt{Minimize}}_{\Upsilon=\{\alpha_1(\nu),\cdots,\alpha_M(\nu)\}}
& ~~~ \frac{1}{\lambda^\ast}-\sigma^2-\sum\limits_{m=1}^M\alpha_m(\nu)h_m(\nu)P \\
\mathtt {Subject \ to} & ~~~ \sum\limits_{m=1}^M\alpha_m(\nu)h_m(\nu)P\leq \frac{1}{\lambda^\ast}-\sigma^2, \\ & ~~~ \alpha_m(\nu) \in \{0,1\}, ~~~ \forall m.
\end{align*}

\begin{table}[htp]
\begin{center}
\caption{\textbf{Algorithm to solve Problem (P3)}} \vspace{0.2cm}
 \hrule
\vspace{0.2cm}
\begin{itemize}
\item[1.] Check whether $\sum\limits_{m=1}^Mh_m(\nu)P\leq \frac{1}{\lambda^\ast}-\sigma^2$. If yes, set $\Upsilon=\{1,\cdots,1\}$ and exit the
algorithm; otherwise, do the following steps.
\item[2.] Given $\epsilon>0$ and $\eta>0$ to control the algorithm accuracy, and set $\mathcal{S}_0=\{0\}$, $\Upsilon_{0,1}=\{0,\cdots,0\}$.
\item[3.] For $i=1:M$
\begin{itemize}
\item[a.] for $j=1:|\mathcal{S}_{i-1}|$
\begin{itemize}
\item[i.] Set $\bar{\mathcal{S}}_i^{(j)}=\mathcal{S}_{i-1}^{(j)}$, $\bar{\mathcal{S}}_i^{(|\mathcal{S}_{i-1}|+j)}=\mathcal{S}_{i-1}^{(j)}+h_i(\nu)P$;
\item[ii.] Set $\bar{\Upsilon}_{i,j}=\Upsilon_{i-1,j}$, $\bar{\Upsilon}_{i,|\mathcal{S}_{i-1}|+j}=\Upsilon_{i-1,j}$, $\bar{\Upsilon}_{i,|\mathcal{S}_{i-1}|+j}^{(i)}=1$;
\end{itemize}
\item[b.] Sort the elements of $\bar{\mathcal{S}}_i$ in a non-decreasing order; adjust $\bar{\Upsilon}_{i,j}$'s accordingly such that each $\bar{\Upsilon}_{i,j}$ indicates the antenna partitions to achieve $\bar{\mathcal{S}}_i^{(j)}$;
\item[c.] Set $n=1$, $\mathcal{S}_i^{(n)}=\{0\}$ and $\Upsilon_{i,n}=\{0,\cdots,0\}$; do for $j=2:|\bar{\mathcal{S}}_i|$
\begin{itemize}
\item[i.] if $\left(1+\frac{\epsilon}{2M}\right) \mathcal{S}_i^{(n)}<\bar{\mathcal{S}}_i^{(j)}\leq \frac{1}{\lambda^\ast}-\sigma^2$, then set $n=n+1$ and $\mathcal{S}_i^{(n)}=\bar{\mathcal{S}}_i^{(j)}$, $\Upsilon_{i,n}=\bar{\Upsilon}_{i,j}$.
\end{itemize}
\item[d.] if $\frac{\frac{1}{\lambda^\ast}-\sigma^2}{1+\eta}\leq \mathcal{S}_i^{(|\mathcal{S}_i|)} \leq \frac{1}{\lambda^\ast}-\sigma^2$, set $\Upsilon=\Upsilon_{i,|\mathcal{S}_i|}$ and exit the
algorithm;

\end{itemize}

\item[4.] Set $\Upsilon=\Upsilon_{M,|\mathcal{S}_M|}$.
\end{itemize}
\vspace{0.2cm} \hrule \label{table1}
\end{center}
\end{table}

For any set $\Omega$, let $|\Omega|$ denote the cardinality of $\Omega$, and $\Omega^{(n)}$ denote the $n$th element in $\Omega$. In Table \ref{table1}, we provide an algorithm to efficiently solve Problem (P3). Note that in Step 1 of the algorithm, all the received power is allocated to the information receiver, i.e., $\Upsilon=\{1,\cdots,1\}$, if $\sum\limits_{m=1}^Mh_m(\nu)P\leq \frac{1}{\lambda^\ast}-\sigma^2$ at a particular fading state. Otherwise, at the $i$th iteration in Step 3a, $\bar{\mathcal{S}}_i$ consists of all the possible values of the total power allocated to the information receiver if only the first $i$ antennas perform antenna switching while the remaining $M-i$ antennas allocate all the received power to the energy receiver, i.e., $\alpha_m(\nu)=0$, $\forall m>i$, and $\bar{\Upsilon}_{i,j}=\{\alpha_1(\nu),\cdots,\alpha_M(\nu)\}$ denotes the antenna switching strategy that achieves the value $\bar{\mathcal{S}}_i^{(j)}$. Steps 3b and 3c aim to eliminate the elements that are close to each other in the set $\mathcal{S}_i$. Finally, the algorithm terminates if the stopping criterion in Step 3d is satisfied.

Note that in this algorithm, if $\epsilon$ is set to zero, then it becomes the exhaustive search method, which has the same complexity order as that of the optimal antenna switching given in Section \ref{sec:Antenna Switching to Approach Optimal Power Splitting}, i.e., $\mathcal{O}(2^M)$. However, the following proposition shows that with a small positive number $\epsilon>0$, the proposed algorithm in Table \ref{table1} has a guaranteed performance as well as a polynomial-time complexity.

\begin{proposition}\label{proposition5}

\ \ \ \ \ \ \

\begin{itemize}
\item[1.] For any $\epsilon>0$, the solution obtained by the algorithm in Table \ref{table1}, $\Upsilon=\{\alpha_1(\nu),\cdots,\alpha_M(\nu)\}$, satisfies
\begin{align}\label{eqn:gap}
\frac{\sum\limits_{m=1}^M\alpha_m^\ast(\nu)h_m(\nu)P}{1+\epsilon} & \leq \sum\limits_{m=1}^M\alpha_m(\nu)h_m(\nu)P \nonumber \\ & \leq \sum\limits_{m=1}^M\alpha_m^\ast(\nu)h_m(\nu)P,
\end{align}where $\{\alpha_1^\ast(\nu),\cdots,\alpha_M^\ast(\nu)\}$ denotes the optimal solution to Problem (P3).
\item[2.] The algorithm in Table \ref{table1} has the worst-case complexity in the order of $\mathcal{O}(M^2)$.
\end{itemize}
\end{proposition}

\begin{proof}
Please refer to Appendix \ref{appendix7}.
\end{proof}

Proposition \ref{proposition5} indicates that (1) the accuracy of the algorithm in Table \ref{table1} can be made arbitrarily high by setting an appropriate value of $\epsilon>0$; and (2) this algorithm has a complexity in the order of $\mathcal{O}(M^2)$, which is significantly lower than $\mathcal{O}(2^M)$ by the exhaustive search.

\subsubsection{The Case With CSIT}

\ \ \ \ \

According to Proposition \ref{proposition4}, in the case of SIMO system with CSIT, the optimal UPS policy for the equivalent SISO system should allocate $\frac{1}{\lambda^\ast}-\sigma^2$ amount of power to the information receiver if the total received power $\sum\limits_{m=1}^Mh_m(\nu)p^\ast(\nu)$ with the optimal transmit power $p^\ast(\nu)$ is larger than $\frac{1}{\lambda^\ast}-\sigma^2$. However, if at any fading state the total received power is less than $\frac{1}{\lambda^\ast}-\sigma^2$, then it should be all allocated to the information receiver (c.f. Fig. \ref{fig5} (a) and (b)). Thus, in the case with CSIT, we can first obtain the optimal transmitter power allocation $p^\ast(\nu)$ for the equivalent SISO system based on Proposition \ref{proposition4}, and then find a pair of $\Phi_{{\rm ID}}(\nu)$ and $\Phi_{{\rm EH}}(\nu)$ for antenna switching such that $\sum\limits_{m\in \Phi_{{\rm ID}}(\nu)}h_m(\nu)p^\ast(\nu)$ is closest to, but no larger than $\frac{1}{\lambda^\ast}-\sigma^2$, similar to the case without CSIT. Therefore, the algorithm proposed in Table \ref{table1} for Problem (P3) (with $P$ replaced by $p^\ast(\nu)$ ) can be applied to the case with CSIT as well to find a low-complexity antenna switching solution.

\subsection{Numerical Results}

\begin{figure}
\begin{center}
\scalebox{0.5}{\includegraphics*{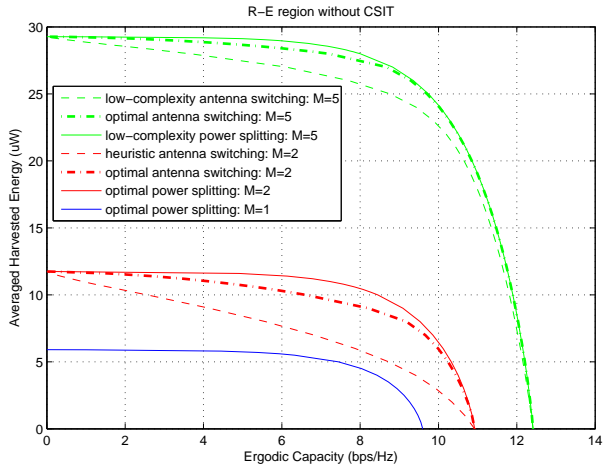}}
\end{center}
\caption{R-E regions of power splitting versus antenna switching for the SIMO system without CSIT.}\label{fig8}
\end{figure}

\begin{figure}
\begin{center}
\scalebox{0.5}{\includegraphics*{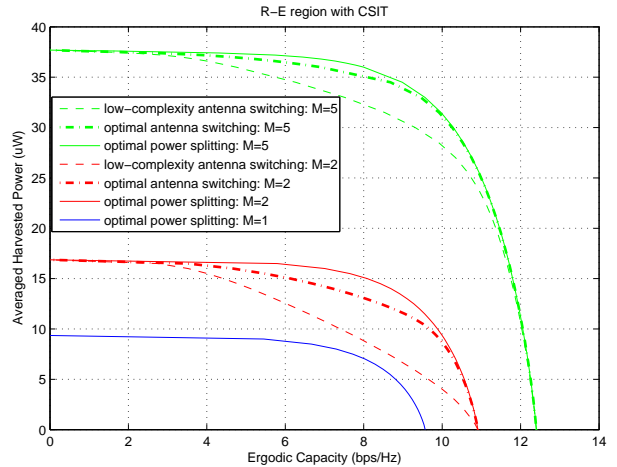}}
\end{center}
\caption{R-E regions of power splitting versus antenna switching for the SIMO system with CSIT.}\label{fig10}
\end{figure}

In this subsection, we provide numerical results to compare the performance of the following three schemes for the SIMO system: the optimal DPS in Section \ref{sec:Optimal Power Splitting in SIMO System}, the optimal antenna switching by exhaustive search in Section \ref{sec:Antenna Switching to Approach Optimal Power Splitting}, and the low-complexity antenna switching in Section \ref{sec:A Heuristic Antenna Switching Scheme with Lower Complexity}. For the proposed algorithm in Table \ref{table1}, both $\epsilon$ and $\eta$ are set as $0.1$. All the parameters for the SIMO setup, e.g., $P_{{\rm peak}}$ and $P_{{\rm avg}}$, are the same as those in the SISO case for Fig. \ref{fig3} in Section \ref{Ergodic Capacity and Harvested Energy Trade-off in Fading Channels}. Furthermore, let $\mv{g}(\nu)=[g_1(\nu),\cdots,g_M(\nu)]^T$ denote the complex channel vector at any fading state $\nu$; then similar to the SISO case, the channel can be modeled as $\mv{g}(\nu)=\sqrt{\frac{K}{K+1}}\hat{\mv{g}}+\sqrt{\frac{1}{K+1}}\tilde{\mv{g}}(\nu)$, where $\hat{\mv{g}}$ is the LOS deterministic component, $\tilde{\mv{g}}(\nu)=[\tilde{g}_1(\nu),\cdots,\tilde{g}_M(\nu)]^T$ denotes the Rayleigh fading component with each element $\tilde{g}_m(\nu)\sim \mathcal{CN}(0,-40{\rm dB})$, and $K$ is the Rician factor set to be $3$. Note that for the LOS component, we use the far-field uniform linear antenna array model \cite{Luo07} with $\hat{\mv{g}}=10^{-4}[1,e^{j\tau},\cdots,e^{j(M-1)\tau}]^T$, where $\tau$ denotes the difference of the phases between two successive receive antennas. Here we set $\tau=-\frac{\pi}{2}$.

Figs. \ref{fig8} and \ref{fig10} compare the achievable R-E regions by the three considered schemes in the SIMO system without versus with CSIT. It is observed that as compared to the case of SISO system with $M=1$, a significantly enlarged R-E region is achieved by using two receiving antennas ($M=2$), even with the low-complexity antenna switching algorithm. It is also observed that as $M$ increases, the performance of the optimal antenna switching by the exhaustive search approaches to that of the optimal UPS. Since antenna switching is a generalization of time switching for the SISO system to the SIMO system, this observation is in sharp contrast to that in Fig. \ref{fig3} where there exists a significant R-E performance loss by time switching as compared to power splitting for the SISO system. More interestingly, as $M$ increases, even the low-complexity antenna switching algorithm is observed to perform very closely to the optimal UPS, which suggests that antenna switching for the SIMO system with a sufficiently large $M$ can be an appealing low-complexity implementation of power splitting in practice.

\section{Conclusion}\label{Concluding Remarks}
This paper studies simultaneous wireless information and power transfer (SWIPT) via the approach of dynamic power splitting (DPS). Under a point-to-point flat-fading SISO
channel setup, we show the optimal power splitting rule at the receiver based on the CSI
to optimize the rate-energy performance
trade-off. When the CSI is also known at the transmitter, the jointly
optimized transmitter
power control and receiver power splitting is derived. The performance of the proposed DPS in the SISO fading channel is compared with that of the existing time switching as well as a performance upper bound obtained by ignoring the practical circuit limitation. Furthermore, we extend the DPS scheme to the SIMO system with multiple receiving antennas and show that a uniform power splitting (UPS) scheme is optimal. We also investigate the practical antenna switching scheme and propose a low-complexity algorithm for it, which can be efficiently implemented to achieve the R-E performance more closely to the optimal UPS as the number of receiving antennas increases.
%
%


\begin{appendix}

\subsection{Proof of Lemma \ref{lemma1}}\label{appendix1}

We consider an infinitesimal interval $\{h(\nu)|\hat{h}\leq h(\nu) \leq \hat{h}+\Delta h\}$, where $\Delta \rightarrow 0$. Since this interval is infinitesimal, we can assume that the value of $h(\nu)$ is constant over this interval, i.e., $h(\nu)=\hat{h}$. Moreover, $f_\nu(h)$ is also a constant denoted by $f_\nu(\hat{h})$ since it is assumed to be a continuous function. As a result, given the constraint pair $(\bar{Q}^a,P_{{\rm avg}}^a)$, the optimal solution can be assumed to be constant within this interval, i.e., $\alpha^a(\nu)=\hat{\alpha}^a$ and $p^a(\nu)=\hat{p}^a$, because the same Karush-Kuhu-Tucker (KKT) conditions hold in the interval. Similarly, given the constraint pair $(\bar{Q}^b,P_{{\rm avg}}^b)$, it follows that $\alpha^a(\nu)=\hat{\alpha}^b$ and $p^a(\nu)=\hat{p}^b$ over this interval. Next, we construct a new solution for Problem (P2) as follows. We divide the interval into two sub-intervals, which have the solution $\alpha^c(\nu)=\hat{\alpha}^a$ and $p^c(\nu)=\hat{p}^a$ corresponding to $\theta$ portion of the interval, and $\alpha^c(\nu)=\hat{\alpha}^b$ and $p^c(\nu)=\hat{p}^b$ for the other $1-\theta$ portion, respectively. It then follows that the average harvested energy in this interval with the new solution is\begin{align*}\Delta Q^c=&(1-\hat{\alpha}^a)\hat{h}\hat{p}^af_\nu(\hat{h})\times \theta \Delta h  \\ & +(1-\hat{\alpha}^b)\hat{h}\hat{p}^bf_\nu(\hat{h})\times (1-\theta) \Delta h .\end{align*}As a result, the average harvested energy over all the fading states can be expressed as\begin{align*}E_\nu[Q^c(\nu)]&=\int\Delta Q^cd\nu=\theta Q^a+(1-\theta)Q^b  \\ &\geq\theta \bar{Q}^a+(1-\theta)\bar{Q}^b,\end{align*}where $Q^\gamma=E_\nu[(1-\alpha^\gamma(\nu))h(\nu)p^\gamma(\nu)]$ with $\gamma\in \{a,b\}$ denotes the average harvested energy by the solution $\{p^\gamma(\nu),\alpha^\gamma(\nu)\}$. Similarly, it can be shown that with the new solution, $E_\nu[r^c(\nu)]\geq \theta E_\nu[r^a(\nu)]+(1-\theta) E_\nu[r^b(\nu)]$ and $E_\nu[p^c(\nu)]\leq \theta P_{{\rm avg}}^a+(1-\theta)P_{{\rm avg}}^b$ can be satisfied. Lemma \ref{lemma1} is thus proved.

\subsection{Proof of Proposition \ref{proposition3}}\label{appendix5}
The derivative of $L_{\nu}^{{\rm
w/o \ CSIT}}(\alpha)$ in (\ref{eqn:1}) with respect to $\alpha$ can be expressed as\begin{align}\frac{\partial L_{\nu}^{{\rm
w/o \ CSIT}}(\alpha)}{\partial \alpha}=\frac{hP}{\alpha hP+\sigma^2}-\lambda hP.\end{align}Since $0\leq \alpha \leq 1$, it follows that \begin{align}\frac{hP}{hP+\sigma^2}-\lambda hP \leq \frac{\partial L_{\nu}^{{\rm
w/o \ CSIT}}(\alpha)}{\partial \alpha} \leq \frac{hP}{\sigma^2}-\lambda hP.\end{align}

If $\frac{hP}{hP+\sigma^2}-\lambda hP\geq 0$, i.e., $h\leq\frac{1}{\lambda P}-\frac{\sigma^2}{P}$, then $\frac{\partial L_{\nu}^{{\rm
w/o \ CSIT}}(\alpha)}{\partial \alpha}\geq 0$ for all $0\leq \alpha\leq 1$. Thus the optimal solution to Problem (\ref{eqn:subproblem3}) is $\alpha^\ast=1$. Otherwise, if $h>\frac{1}{\lambda P}-\frac{\sigma^2}{P}$, the maximum of $L_{\nu}^{{\rm
w/o \ CSIT}}(\alpha)$ is achieved when $\frac{\partial L_{\nu}^{{\rm
w/o \ CSIT}}(\alpha)}{\partial \alpha}=0$, i.e., $\alpha^\ast=\frac{1}{\lambda hP}-\frac{\sigma^2}{hP}$. Proposition \ref{proposition3} is thus proved.

\subsection{Proof of Proposition \ref{proposition4}}\label{appendix6}

The derivative of $L_{\nu}^{{\rm
with \ CSIT}}(p,\alpha)$ given in (\ref{eqn:2}) with respect to $\alpha$ can be expressed as\begin{align}\frac{\partial L_{\nu}^{{\rm
with \ CSIT}}(p,\alpha)}{\partial \alpha}=\frac{hp}{\alpha hp+\sigma^2}-\lambda hp.\end{align}
For any given $p\in [0,P_{{\rm peak}}]$, since $0\leq \alpha \leq 1$, it follows that \begin{align}\frac{hp}{hp+\sigma^2}-\lambda hp \leq \frac{\partial L_{\nu}^{{\rm
with \ CSIT}}(p,\alpha)}{\partial \alpha} \leq \frac{hp}{\sigma^2}-\lambda hp.\end{align}It can be shown that if $\lambda \geq \frac{1}{\sigma^2}$, it follows that $\frac{\partial L_{\nu}^{{\rm
with \ CSIT}}(p,\alpha)}{\partial \alpha}\leq 0$, $\forall p$. In this case, for all the fading states we have $\alpha^\ast=0$, which implies that Problem (\ref{eqn:subproblem4}) is not feasible. As a result, in the following we only consider the case of $\lambda<\frac{1}{\sigma^2}$.

Define $S_1$ and $S_2$ as follows:\begin{align}& S_1=\left\{p\bigg|\frac{hp}{hp+\sigma^2}-\lambda hp \geq 0, \ 0\leq p \leq P_{{\rm peak}}\right\}, \\ & S_2=\left\{p\bigg|\frac{hp}{hp+\sigma^2}-\lambda hp<0, \ 0\leq p \leq P_{{\rm peak}}\right\}.\end{align}To be specific, if $\frac{1}{\lambda h}-\frac{\sigma^2}{h}\leq P_{{\rm peak}}$, i.e., $h\geq \frac{1}{\lambda P_{{\rm peak}}}-\frac{\sigma^2}{P_{{\rm peak}}}$, it follows that\begin{align}& S_1=\left\{p\bigg|0\leq p \leq \frac{1}{\lambda h}-\frac{\sigma^2}{h}\right\}, \label{eqn:s1case1} \\ & S_2=\left\{p\bigg|\frac{1}{\lambda h}-\frac{\sigma^2}{h} < p \leq P_{{\rm peak}}\right\}. \label{eqn:s2case1}\end{align}Otherwise, we have
\begin{align}& S_1=\left\{p\bigg|0\leq p \leq P_{{\rm peak}}\right\}, \label{eqn:s1case2} \\ & S_2=\emptyset. \label{eqn:s2case2}\end{align}

It can be shown that if $p\in S_1$, then \begin{align}\frac{\partial L_{\nu}^{{\rm
with \ CSIT}}(p,\alpha)}{\partial \alpha} \geq \frac{hp}{hp+\sigma^2}-\lambda hp \geq 0, \ \forall \alpha.\end{align}In this case, $L_{\nu}^{{\rm
with \ CSIT}}(p,\alpha)$ is a monotonically increasing function of $\alpha$, and thus the optimal power splitting ratio is $\alpha^\ast=1$. If $p\in S_2$, then we have\begin{align}\frac{\partial L_{\nu}^{{\rm
with \ CSIT}}(p,\alpha)}{\partial \alpha}=0 \ \Rightarrow \ \alpha^\ast=\frac{1}{\lambda hp}-\frac{\sigma^2}{hp}.\end{align}

To summarize, we have
\begin{align}
& L_{\nu}^{{\rm with \ CSIT}}(p,\alpha^\ast) \nonumber \\ = &\left\{\begin{array}{ll}\log\left(1+\frac{hp}{\sigma^2}\right)-\beta p, & {\rm if} ~ p\in S_1, \\ \log\frac{1}{\lambda \sigma^2}+\lambda hp-\beta p+\lambda \sigma^2-1, & {\rm if} ~ p\in S_2.\end{array}\right.
\end{align}

To find the optimal power allocation $p^\ast$ given any channel power $h$, we need to compare the optimal values of the following two subproblems.\begin{align*}\mathrm{(P2.1)}:~\mathop{\mathtt{Maximize}}\limits_{p}
& ~~~ \log\left(1+\frac{hp}{\sigma^2}\right)-\beta p \\
\mathtt {Subject \ to} & ~~~ p\in S_1,
\end{align*}\begin{align*}\mathrm{(P2.2)}:~\mathop{\mathtt{Maximize}}\limits_{p}
& ~~~ \log\frac{1}{\lambda \sigma^2}+\lambda hp-\beta p+\lambda \sigma^2-1 \\
\mathtt {Subject \ to} & ~~~ p\in S_2.
\end{align*}Since the expressions of $S_1$ and $S_2$ depend on the relationship between $h$ and $\frac{1}{\lambda P_{{\rm peak}}}-\frac{\sigma^2}{P_{{\rm peak}}}$, in the following we solve Problems (P2.1) and (P2.2) in two different cases.

\begin{itemize}
\item[1)] {\bf Case I}: $h\geq \frac{1}{\lambda P_{{\rm peak}}}-\frac{\sigma^2}{P_{{\rm peak}}}$
\end{itemize}

In this case, $S_1$ and $S_2$ are expressed in (\ref{eqn:s1case1}) and (\ref{eqn:s2case1}), respectively. Then the optimal solution to Problem (P2.1) can be expressed as\begin{align}\label{eqn:case1}p=\left\{\begin{array}{ll}\frac{1}{\lambda h}-\frac{\sigma^2}{h}, & {\rm if} ~ h\geq \psi, \\ \left(\frac{1}{\beta}-\frac{\sigma^2}{h}\right)^+, & {\rm if} ~ \frac{1}{\lambda P_{{\rm peak}}}-\frac{\sigma^2}{P_{{\rm peak}}} \leq h < \psi,\end{array}\right.\end{align}where $\psi=\max\{\frac{\beta}{\lambda},\frac{1}{\lambda P_{{\rm peak}}}-\frac{\sigma^2}{P_{{\rm peak}}}\}$, and $(x)^+=\max \{0,x\}$. Furthermore, the optimal solution to Problem (P2.2) can be obtained as\begin{align}\label{eqn:case2}p=\left\{\begin{array}{ll}P_{{\rm peak}}, & {\rm if} ~ h\geq \psi, \\\frac{1}{\lambda h}-\frac{\sigma^2}{h}, & {\rm if} ~ \frac{1}{\lambda P_{{\rm peak}}}-\frac{\sigma^2}{P_{{\rm peak}}}\leq h < \psi.\end{array}\right.\end{align}

Since the expressions of (\ref{eqn:case1}) and (\ref{eqn:case2}) depend on the relationship between $\frac{\beta}{\lambda}$ and $\frac{1}{\lambda P_{{\rm peak}}}-\frac{\sigma^2}{P_{{\rm peak}}}$, in the following we further discuss two subcases.

\begin{itemize}
\item {\bf Subcase I-i}: $\frac{1}{\lambda P_{{\rm peak}}}-\frac{\sigma^2}{P_{{\rm peak}}}\leq \frac{\beta}{\lambda}$
\end{itemize}

In this subcase, $\psi=\frac{\beta}{\lambda}$. It can be observed from (\ref{eqn:case1}) and (\ref{eqn:case2}) that if $\psi=\frac{\beta}{\lambda}$, the optimal power solution to Problem (P2.1) is $\frac{1}{\lambda h}-\frac{\sigma^2}{h}$, $\frac{1}{\beta}-\frac{\sigma^2}{h}$ or $0$, while that to Problem (P2.2) is $P_{{\rm peak}}$ or $\frac{1}{\lambda h}-\frac{\sigma^2}{h}$, depending on the value of $h$. Therefore, three cases exist when $\frac{1}{\lambda P_{{\rm peak}}}-\frac{\sigma^2}{P_{{\rm peak}}}\leq \frac{\beta}{\lambda}$, discussed as follows.

The first case is $h\geq\frac{\beta}{\lambda}$, for which the difference between the optimal values of Problems (P2.1) and (P2.2) can be expressed as\begin{align}\label{eqn:subcase}d_1=&\left[\log\left(1+\frac{hp}{\sigma^2}\right) -\beta p\right]\bigg|_{p=\frac{1}{\lambda h}-\frac{\sigma^2}{h}} \nonumber \\ & -\left[\log \frac{1}{\lambda \sigma^2}+\lambda hp-\beta p+\lambda \sigma^2-1\right]\bigg|_{p=P_{{\rm peak}}} \nonumber \\ =&\left(\log \frac{1}{\lambda \sigma^2}-\frac{\beta}{\lambda h}+\frac{\beta \sigma^2}{h} \right) \nonumber \\ & -\left(\log \frac{1}{\lambda \sigma^2}+\lambda hP_{{\rm peak}}-\beta P_{{\rm peak}}+\lambda \sigma^2-1 \right) \nonumber \\ =& (\lambda h-\beta)\left(\frac{1}{\lambda h}-\frac{\sigma^2}{h}-P_{{\rm peak}}\right)<0.\end{align}Therefore, if $h\geq \frac{\beta}{\lambda}$, the optimal value of Problem (P2.2) is always larger than that of Problem (P2.1), and the optimal solution to Problem (\ref{eqn:subproblem4}) is $p^\ast=P_{{\rm peak}}$ and $\alpha^\ast=\frac{1}{\lambda hP_{{\rm peak}}}-\frac{\sigma^2}{hP_{{\rm peak}}}$.


The second case is $\max(\frac{1}{\lambda P_{{\rm peak}}}-\frac{\sigma^2}{P_{{\rm peak}}},\beta \sigma^2)\leq h<\frac{\beta}{\lambda}$, for which the difference between the optimal values of Problems (P2.1) and (P2.2) can be expressed as\begin{align}d_2=& \left[\log\left(1+\frac{hp}{\sigma^2}\right)-\beta p\right]\bigg|_{p=\frac{1}{\beta}-\frac{\sigma^2}{h}} \nonumber \\ & -\left[\log \frac{1}{\lambda \sigma^2}+\lambda hp-\beta p+\lambda \sigma^2-1\right]\bigg|_{p=\frac{1}{\lambda h}-\frac{\sigma^2}{h}} \nonumber \\ = & \left(\log\frac{h}{\beta \sigma^2}-1+\frac{\beta \sigma^2}{h}\right)-\left(\log\frac{1}{\lambda \sigma^2}-\frac{\beta}{\lambda h}+\frac{\beta \sigma^2}{h}\right) \nonumber \\ = & \log\frac{\lambda h}{\beta}+\frac{\beta}{\lambda h}-1.\end{align}It can be shown that the function $f(x)=\log x+\frac{1}{x}-1$ is a monotonically decreasing function in the interval $(0,1]$. Moreover, $\frac{\lambda h}{\beta}<1$. It then follows that \begin{align}d_2=\left[\log x+\frac{1}{x}-1\right]\bigg|_{x=\frac{\lambda h}{\beta}}>\left[\log x+\frac{1}{x}-1\right]\bigg|_{x=1}=0.\end{align}Thus, for this case the optimal value of Problem (P2.1) is always larger than that of Problem (P2.2), and the optimal solution to Problem (\ref{eqn:subproblem4}) is $p^\ast=\frac{1}{\beta}-\frac{\sigma^2}{h}$ and $\alpha^\ast=1$.

The third case is $\frac{1}{\lambda P_{{\rm peak}}}-\frac{\sigma^2}{P_{{\rm peak}}}\leq h< \beta \sigma^2$ (if $\frac{1}{\lambda P_{{\rm peak}}}-\frac{\sigma^2}{P_{{\rm peak}}}<\beta \sigma^2$), for which the difference between the optimal values of Problems (P2.1) and (P2.2) can be expressed as\begin{align}d_3=& \left[\log\left(1+\frac{hp}{\sigma^2}\right)-\beta p\right]\bigg|_{p=0} \nonumber \\ & -\left[\log \frac{1}{\lambda \sigma^2}+\lambda hp-\beta p+\lambda \sigma^2-1\right]\bigg|_{p=\frac{1}{\lambda h}-\frac{\sigma^2}{h}} \nonumber \\ = &-\left(\log\frac{1}{\lambda \sigma^2}-\frac{\beta}{\lambda h}+\frac{\beta \sigma^2}{h}\right) \nonumber \\ \overset{(a)}{\geq} & -\left(\log\frac{1}{\lambda \sigma^2}-\frac{1}{\lambda \sigma^2}+1\right) \nonumber \\ \overset{(b)}{\geq} & 0,\end{align}where $(a)$ is due to the fact that the function on the left hand side is a decreasing function in the interval of $h\in [\frac{1}{\lambda P_{{\rm peak}}}-\frac{\sigma^2}{P_{{\rm peak}}},\beta \sigma^2)$ if $\lambda<\frac{1}{\sigma^2}$, while $(b)$ is due to that $f(x)=-\log x+x-1$ is an increasing function if $x\geq 1$, and thus $f(x=\frac{1}{\lambda\sigma^2})\geq f(x=1)=0$. As a result, for this case the optimal value of Problem (P2.1) is larger than that of Problem (P2.2), and the optimal solution to Problem (\ref{eqn:subproblem4}) is thus $p^\ast=0$, $\alpha^\ast=1$.

\begin{itemize}
\item {\bf Subcase I-ii}: $\frac{1}{\lambda P_{{\rm peak}}}-\frac{\sigma^2}{P_{{\rm peak}}}> \frac{\beta}{\lambda}$
\end{itemize}

In this subcase, $\psi=\frac{1}{\lambda P_{{\rm peak}}}-\frac{\sigma^2}{P_{{\rm peak}}}$. It can be observed from (\ref{eqn:case1}) and (\ref{eqn:case2}) that if $\psi=\frac{1}{\lambda P_{{\rm peak}}}-\frac{\sigma^2}{P_{{\rm peak}}}$, the optimal power solution to Problem (P2.1) is $\frac{1}{\lambda h}-\frac{\sigma^2}{h}$, and that to Problem (P2.2) is $P_{{\rm peak}}$, and the difference between the optimal values of Problems (P2.1) and (P2.2) can be expressed as (\ref{eqn:subcase}). Thus, if $\frac{1}{\lambda P_{{\rm peak}}}-\frac{\sigma^2}{P_{{\rm peak}}}> \frac{\beta}{\lambda}$, the optimal solution to Problem (\ref{eqn:subproblem4}) is given by $p^\ast=P_{{\rm peak}}$ and $\alpha^\ast=\frac{1}{\lambda hP_{{\rm peak}}}-\frac{\sigma^2}{hP_{{\rm peak}}}$.

\begin{itemize}
\item[2)] {\bf Case II}: $h<\frac{1}{\lambda P_{{\rm peak}}}-\frac{\sigma^2}{P_{{\rm peak}}}$
\end{itemize}

In this case, $S_1$ and $S_2$ are expressed in (\ref{eqn:s1case2}) and (\ref{eqn:s2case2}), respectively. Since $S_2=\emptyset$, the optimal power splitting ratio to Problem (\ref{eqn:subproblem4}) is $\alpha^\ast=1$. Moreover, the optimal power allocation is only determined by Problem (P2.1), which can be expressed as $p^\ast=\left[\frac{1}{\beta}-\frac{\sigma^2}{h}\right]^{P_{{\rm peak}}}_0$, where $[x]_a^b=\max(\min(x,b),a)$.

By combining the above results, Proposition \ref{proposition4} is thus proved.
\subsection{Proof of Proposition \ref{proposition5}}\label{appendix7}

It is observed in Table \ref{table1} that at any iteration $i$, if any element in $\mathcal{S}_i$ does not exceed its previous element by a ratio of $\frac{\epsilon}{2M}$, it will not be included in the same set. As a result, each iteration introduces a multiplicative error factor of at most $\frac{\epsilon}{2M}$. In the worst case, it is then guaranteed that
\begin{align}
\frac{\sum\limits_{m=1}^M\alpha_m^\ast(\nu)h_m(\nu)P}{\sum\limits_{m=1}^M\alpha_m(\nu)h_m(\nu)P}\leq \left(1+\frac{\epsilon}{2M}\right)^M\leq 1+\epsilon.
\end{align}The first part of Proposition \ref{proposition5} is thus proved.

Next, at each iteration $i$, let $s_i^{{\rm min}}$ denote the smallest positive element in the set $\mathcal{S}_i$. Since each element in $\mathcal{S}_i$ is at least $\frac{\epsilon}{2M}$ times larger than its previous element, it follows that
\begin{align}
s_i^{{\rm min}}\left(1+\frac{\epsilon}{2M}\right)^{|\mathcal{S}_i|-2}\leq \frac{1}{\lambda^\ast}-\sigma^2.
\end{align}In other words, by defining $\tau_i=\frac{\frac{1}{\lambda^\ast}-\sigma^2}{s_i^{{\rm min}}}$, then at each iteration $i$, the size of $\mathcal{S}_i$ must satisfy
\begin{align}
|\mathcal{S}_i| & \leq 2+\log_{(1+\frac{\epsilon}{2M})}\tau_i \nonumber \\
& =2+\frac{\log \tau_i}{\log\left(1+\frac{\epsilon}{2M}\right)} \nonumber \\
& \overset{(a)}{\leq} 2+\frac{4M\log \tau_i}{\epsilon}, \label{eqn:size}
\end{align}where $(a)$ is due to $f(x)=\log(1+x)-\frac{x}{2}>0$ when $0<x\leq 1$, and $x=\frac{\epsilon}{2M}\ll 1$.

It is observed from (\ref{eqn:size}) that all the sets $\mathcal{S}_i$'s with $1\leq i \leq M$ have their sizes linearly growing with $M$; thus, since in Table \ref{table1} the algorithm has at most $M$ iterations, its complexity is in the order of $\mathcal{O}(M^2)$ for the worst case. The second part of Proposition \ref{proposition5} is thus proved.

\end{appendix}

%
\end{document}